\documentclass{siamltex}

\usepackage{amsfonts, amsmath, amssymb, psfrag, graphicx, bbm, mathrsfs,
enumitem,calc}

\usepackage{soul}

\usepackage[colorlinks=true]{hyperref}
\hypersetup{urlcolor=blue, citecolor=blue, linkcolor=blue}

\newtheorem{example}[theorem]{Example}                         
\newtheorem{remark}[theorem]{Remark}

\newtheorem{assumption}[theorem]{Assumption}

\def\fin{\ifmmode{\Large$\diamond$}\else{\unskip\nobreak\hfil
    \penalty50\hskip1em\null\nobreak\hfil{\Large$\diamond$}
    \parfillskip=0pt\finalhyphendemerits=0\endgraf}\fi}

\def\be#1#2\ee{\begin{equation}\label{eq:#1}#2\end{equation}}
\def\req#1{{\rm(\ref{eq:#1})}}
\def\bdm  {\begin{displaymath}}
  \def\edm  {\end{displaymath}}
\def\bdmal{\begin{displaymath}\begin{aligned}}
    \def\edmal{\end{aligned}\end{displaymath}}

\def\thsp{\hspace*{0.1ex}}

\mathchardef\PhiG="0108
\mathchardef\Sigma="0106
\mathchardef\Delta="0101
\newcommand{\dI}{{\cal I}_\square}
\newcommand{\Loewnergeq}{\succcurlyeq}
\newcommand{\Loewnergr}{\succ}

\renewcommand{\L}{{\mathscr L}}

\newcommand{\R}{{\mathord{\mathbb R}}}

\newcommand{\C}{{\mathord{\mathbb C}}}
\newcommand{\E}{{\mathord{\mathbb E}}}

\newcommand{\Normal}{{\cal N}}

\renewcommand{\L}{{\mathscr L}}

\newcommand{\norm}[1]{\|#1\|}

\newcommand{\scalp}[1]{\langle\,#1\,\rangle}

\newcommand{\Real}{{\rm Re\,}}

\newcommand{\rmd}{\,\mathrm{d}}
\newcommand{\ds}{\rmd s}

\newcommand{\dt}{\rmd t}

\newcommand{\dtau}{\rmd \tau}

\newcommand{\dW}{\rmd W}
\newcommand{\dWtilde}{\rmd\Wtilde}
\newcommand{\rmi}{\mathrm{i}}
\newcommand{\eps}{\varepsilon}

\def\req#1{{\rm(\ref{eq:#1})}}

\makeatletter
\newcommand{\dupdots}{\mathinner{\mkern1mu\raise\p@
    \vbox{\kern7\p@\hbox{.}}\mkern2mu
    \raise4\p@\hbox{.}\mkern2mu\raise7\p@\hbox{.}\mkern1mu}}
\makeatother
\makeatletter
\newcount\c@MaxMatrixCols \c@MaxMatrixCols=10
\newenvironment{cmatrixc}{%
  \hskip -\arraycolsep\array{*\c@MaxMatrixCols c}%
}{%
  \endarray \hskip -\arraycolsep
}
\makeatother
\newenvironment{cmatrix}{\left[\cmatrixc}{\endmatrix\right]}

 {\endMakeFramed}

\newcommand{\phitrev}{\overset{\leftarrow}{\varphi}{}}

\newcommand{\Wtilde}{{\widetilde W}}

\newcommand{\phihat}{{\widehat\varphi}}
\newcommand{\psihat}{{\widehat\psi}}

\newcommand{\gammahat}{{\widehat\gamma}}
\newcommand{\rhat}{{\widehat r}}

\makeatletter
\renewcommand\@biblabel[1]{#1.}
\makeatother

\title{The second fluctuation-dissipation theorem for the generalized
Langevin equation\thanks{The research 
    leading to this work has been done within
    the Collaborative Research Center TRR~146; corresponding funding 
    by the DFG is gratefully acknowledged.}}
\author{Martin Hanke\thanks{Institut f\"ur Mathematik, Johannes
    Gutenberg-Universit\"at Mainz, 55099 Mainz, Germany
    ({\tt hanke@math.uni-mainz.de})}}

\begin{document}
\sloppy
\maketitle

\vspace*{1ex}

\begin{abstract}
Necessary and sufficient conditions are presented for the existence of 
(second order) stationary solutions of the generalized Langevin equation
under appropriate assumptions on the associated memory kernel.
When this stochastic equation is formulated as an initial value problem, 
then it is shown that the solution approaches a stationary process as time 
goes to infinity, whenever the fluctuating force term is taken to be a 
combination of a white noise process and a mean-square continuous centered 
stationary Gaussian process (which may be correlated with each other). 
The limiting process can be any centered stationary Gaussian process 
with sufficiently smooth spectral density.
On the other hand, the solution itself is only stationary when the fluctuating
force satisfies a certain fluctuation-dissipation relation, 
and this stationary solution is uniquely specified by its covariance.
\end{abstract}

\begin{keywords}
Generalized Langevin equation, stationary stochastic processes
\end{keywords}

\begin{AMS}
{\sc 60G10, 60G15, 60G20}
\end{AMS}


%

\section{Introduction}
\label{Sec:Introduction}
In statistical physics the stochastic Ornstein-Uhlenbeck differential equation
\be{OU}
   \dot{V}(t) \,=\, -DV(t) \,+\, G\dot{W}(t)\,, \qquad t>0\,,
\ee
where $D$ and $G$ are given real $d\times d$ matrices with $-D$ being stable,
and $\dot{W}$ is a $d$-dimensional white-noise process,
is known as \emph{Langevin equation}.
It is used as a phenomenological model for the
motion of a Brownian particle with mass $m>0$,
immersed in a fluid of microparticles: in this case $d=3$, say, is the space
dimension and $V\in\R^3$ the particle's velocity; 
$-DV$ is the friction term with
$D=\gamma I$, where $\gamma>0$ is the friction coefficient 
and $I\in\R^{3\times 3}$ the identity matrix, 
and $F=G\dot{W}$ is the fluctuating (specific) force. 

A well-known theorem, cf., e.g., Pavliotis~\cite[Proposition~3.5]{Pavl14}, 
states that the differential equation~\req{OU} has a
stationary centered mean-square continuous solution $V$ 
(also called wide-sense stationary or second order stationary) with covariance
\bdm
   \Sigma \,=\, \E\bigl(V(t)V(t)^*\bigr)\,,
\edm
if and only if $\Sigma$ solves the Lyapunov matrix equation
\be{Lyapunov}
   D\Sigma \,+\, \Sigma D^* \,=\, GG^*\,.
\ee
In the context of the Brownian particle the equipartition theorem
prescribes that the average energy of this particle in equilibrium is given by
$3/(2\beta) = 3k_BT/2$, where $T$ is the temperature of the system and 
$k_B$ is Boltzmann's constant, which means that
\be{equi}
   \Sigma \,=\, \frac{1}{\beta m}\, I\,.
\ee
Accordingly, \req{Lyapunov} yields
\be{fdtLE}
   GG^* \,=\, \frac{2\gamma}{\beta m}\, I\,.
\ee

The relation~\req{fdtLE} is known as 
\emph{fluctuation-dissipation theorem}; 
quoting Pottier~\cite[p.~240]{Pott10}, 
this name 
expresses the fact that the friction force and the fluctuating force 
represent two aspects of the same physical phenomenon, namely the collisions
of the Brownian particle with the molecules of the fluid which surround it.
Fluctuation-dissipation theorems like this one form the basis
of linear response theory in statistical physics.

The fluctuation-dissipation theorem was investigated in detail by 
Kubo~\cite{Kubo66}. Following his terminology, the connection~\req{fdtLE}
is nowadays often referred to as the 
\emph{second} fluctuation-dissipation theorem.
In fact, the focus in \cite{Kubo66} is on the 
\emph{generalized Langevin equation}
\be{GLE-Kubo}
   \dot{V}(t) \,=\, - \int_0^{t} \gamma(t-s) V(t)\ds \,+\, F(t)\,, \qquad
   t>0\,,
\ee
where the memory kernel $\gamma$ may be scalar or matrix-valued,
and $F$ is a fluctuating force, again. This model
is considered to be more appropriate for the dynamics of the above particle
when its mass is similar to the mass of the fluid molecules, so that
the relevant time scales associated with the motion of the different types
of particles overlap. Under the assumptions
that the fluctuating force is uncorrelated to the initial velocity 
of the particle at time $t=0$, that the velocity is a 
stationary stochastic process\footnote{All stationary processes
considered in this paper are centered; therefore, here and below,
the notion ``centered'' is omitted for the ease of reading.}, and 
that the fluctuating force is a stationary process
with autocorrelation function
\be{F-Kubo} 
   C_F(t) \,=\, \E\bigl(F(s+t)F(s)^*\bigr) \,=\, \frac{1}{\beta m}\,\gamma(t)\,,
   \qquad
   s,t\geq 0\,,
\ee
i.e., colored noise instead of the white noise process in \req{OU}, 
Kubo shows that the equipartition theorem is satisfied, i.e., 
that the covariance $\Sigma$ of the velocity is given by \req{equi}.
Again, he refers to this as the second fluctuation-dissipation theorem.

Kubo's assumptions are somewhat restrictive, and in particular, the requirement
that force and initial velocity are uncorrelated is sometimes questioned;
cf., e.g., Hohenegger and McKinley~\cite[Section~2.1]{HoMc18}.
Furthermore, important questions remain open: Is it clear that the assumptions 
made by Kubo are not inherently contradictory? If not, are there other 
fluctuating forces which lead to the same stationary solution of 
\req{GLE-Kubo}? And are there fluctuating forces which give rise to
other stationary solutions of \req{GLE-Kubo}, 
that do or do not satisfy the equipartition theorem? 
As a consequence, following Jung and Schmid~\cite{JuSc21}, 
Schilling~\cite{Schi24}, and others, 
it appears to be more appropriate to refer to \req{F-Kubo} as a 
\emph{fluctuation-dissipation relation}, rather than a theorem,
despite the fact that this relation is satisfied in 
relevant bottum-up modeling scenarios;
compare, e.g., the recent survey by Schilling~\cite{Schi22},

%

On the other hand, the interpretation of 
Kubo's second fluctuation-dissipation theorem (or relation) 
in the physics community goes far beyond the original statement
in \cite{Kubo66}, as is documented by representative quotes from
the pertinent literature:
\begin{quote}
\begin{itemize}[leftmargin=-1ex]
\item[-] ``\emph{To generate the correct equilibrium statistics for the CG 
model, the random noise has to obey the second fluctuation dissipation 
theorem}'' (from \cite{LBL16});
\item[-]
``\emph{The fluctuation-dissipation theorem states that equilibration to a 
temperature $T$ requires that the two-time correlation of $F$ and $\gamma$
be related as ...}'' (from \cite{BaBo13});
\item[-]
``\emph{the fluctuation-dissipation theorem dictates that the stationary 
autocorrelation function of the random force, temperature $T$, and
the memory kernel are related ...}'' (from \cite{Goyc09}). 
\end{itemize}
\end{quote}

Partial answers to the aforementioned questions 
have recently been provided in \cite[Appendix~A]{JuSc21}.
The purpose of this paper is to formulate a rigorous version of the second
fluctuation-dissipation theorem for the generalized Langevin equation,
which settles sufficiency \emph{and} necessity 
of the fluctuation-dissipation relation~\req{F-Kubo}
for the validity of the equipartition theorem.
The equation to be studied in the sequel is slightly
more general than \req{GLE-Kubo}, as it comes with an 
instantaneous drift term. This adds quite a bit of technical 
burden to the analysis, but it allows for more freedom 
in the choice of phenomenologial models for simulation purposes, 
and the final result improves upon previous work for the very same form
of the generalized Langevin equation.

In the physical context one might expect that for given initial data $V_0$
at time $t=0$ the solution of the generalized Langevin equation~\req{GLE-Kubo} 
will not automatically be a stationary process, but rather that the solution
approaches a stationary one as $t\to+\infty$; compare, e.g.,
Shin, Kim, Talkner, and Lee~\cite{SKTL10}. It is reasonable to assume
-- and will be verified in Theorem~\ref{Thm:GLEinf} below --
that the stationary limit will be a solution of a generalized Langevin equation
similar to \req{GLE-Kubo}, but with infinite time horizon, i.e., where the
integration starts at minus infinity. Another fundamental issue with the
model~\req{GLE-Kubo} is that even when its solution is
stationary \emph{per se}, the \emph{joint process} $(V,F)$ fails to be 
stationary in general, as has been shown in \cite{HoMc18}. This is
irritating, as in equilibrium physics one would expect that the correlations 
between the forces and the velocities in question are also independent of time.
It is therefore interesting to investigate in more detail the 
infinite time horizon generalized Langevin equation 
and its stationary solutions. It turns out
-- as already observed by Mitterwallner et al~\cite{MSDRN20} -- 
that the family of stationary solutions of this variant of the
generalized Langevin equation is much richer than for \req{GLE-Kubo}.

The outline of this paper is as follows. Section~\ref{Sec:Existence} presents
the main results concerning the generalized Langevin equation formulated
as an initial value problem. First, 
the notion and existence of solutions is settled under general assumptions
on the fluctuating force, and then a corresponding
fluctuation-dissipation theorem is formulated (Theorem~\ref{Thm:fdt}):
It states that when the fluctuating force is a
stationary Gaussian process with a given structure as specified 
in Assumption~\ref{Ass:F} below, then a stationary solution 
of the generalized Langevin equation with a prescribed covariance
matrix exists, if and only if the fluctuating force satisfies an appropriate 
fluctuation-dissipation relation.
For the particular equation~\req{GLE-Kubo} this shows that under very mild 
assumptions on the memory kernel the 
fluctuation-dissipation relation~\req{F-Kubo} is indeed sufficient \emph{and}
necessary to satisfy the equipartition theorem;
so this ultimately justifies the above quotations from the literature.
The general form of the result is discussed in great detail in 
Section~\ref{Sec:discussion}.
Then, Section~\ref{Sec:GLEinf} treats the generalized Langevin equation 
with infinite time horizon and, again, investigates the associated stationary 
solutions. In order to simplify the reading most of the proofs of the 
theoretical results have been postponed to three appendices.

\section{Existence of strong solutions of the generalized Langevin equation
and the fluctuation-dissipation theorem}
\label{Sec:Existence}
The particular form of the generalized Langevin equation, 
which will be studied in the sequel, is
\be{GLE}
   \dot{V}(t) \,=\, -DV(t) \,-\, \int_0^t \gamma(t-s)V(s)\ds \,+\, F(t)\,,
   \qquad
   V(0) \,=\, V_0\,,
\ee
where the $d$-dimensional process $V$ is sought, the memory kernel $\gamma$ 
takes values in $\R^{d\times d}$, and $F$ is the fluctuating force term.
Compared to \req{GLE-Kubo}, the model~\req{GLE} comes with 
an additional term $-DV$, where $D\in\R^{d\times d}$.
If this term is modeling some instantaneous friction, as in \req{OU}, then
$D$ is symmetric positive definite; but there are also instances, 
in which $V$ represents some general physical quantities and 
\req{GLE} is derived from first principles as, e.g., in \cite{Mori65,Zwan01}, 
where $D$ is a skew-symmetric matrix. And finally,
when \req{GLE} is used as a phenomenological model of some observational data, 
then it may make sense to allow for general matrices $D$ in \req{GLE}; 
see Section~\ref{Sec:discussion} below.

In general, the initial data $V_0$ and the fluctuating force are
$d$-dimensional real-valued stochastic quantities.
In the sequel $V_0$ is a centered Gaussian random variable
and $F$ is a Gaussian process;
see Pavliotis~\cite{Pavl14} and Cram\'{e}r and Leadbetter~\cite{CL67}
as general references concerning Gaussian processes.
In the stochastic differential equation~\req{GLE} 
the symbol $\dot{V}$ is used as a shorthand notation
with the meaning that a stochastic process
$V$ is a strong solution of \req{GLE}, if and only if $V(t)$
almost surely\footnote{All identities of the individual stochastic
processes considered in this work are satisfied almost surely. 
The corresponding phrase is omitted for the ease of reading.} 
satisfies the integrated identity
\bdm
   V(t) \,=\, V_0 \,-\, D\int_0^t V(t)\dt 
                  \,-\, \int_0^t \int_0^\tau \gamma(\tau-s)V(s)\ds\dtau
                  \,+\, \int_0^t F(\tau)\dtau
\edm
for all $t\in\R^+$. 
In contrast, the prime notation $V'$ is reserved for stochastic processes $V$,
for which the derivative is known to exist in the mean-square sense.
Whether $\dot{V}$ can be replaced by $V'$ in \req{GLE} depends on the
regularity of the fluctuating force; see the proof of Theorem~\ref{Thm:GLE}.


Following Riedle~\cite{Ried03}, who applied the analytical theory of
Volterra integro-differential equations developed in the late 20th century
(compare the monograph of Gripenberg, Londen, and Staffans~\cite{GLS90})
to the generalized Langevin equation with infinite time horizon, 
the deterministic integro-differential equation
\be{DiffResolv}
   r'(t) \,=\, -Dr(t) \,-\, \int_0^t \gamma(t-s)r(s)\ds\,, \quad t\geq 0\,, 
   \qquad 
   r(0)\,=\, I\,,
\ee
where $I$ is the $d\times d$ identity matrix, is of fundamental relevance 
for this work.
When the memory kernel belongs to $L^1_{\rm loc}(\R^+_0)$, i.e., if $\gamma$
is absolutely integrable on every bounded time interval, 
then \req{DiffResolv} has a unique solution $r$ (with values in $\R^{d\times d}$),
cf.~~\cite[Theorem~3.3.1]{GLS90}, the so-called \emph{differential resolvent}.
This solution is absolutely continuous, and \req{DiffResolv} holds true 
almost everywhere in $\R^+$. 
In the present context, where the right-hand side of \req{DiffResolv} is 
continuous in $t$, the (classical) derivative of $r$ is a continuous 
function and \req{DiffResolv} is valid for all $t\in\R^+$.
Moreover (see \cite{GLS90} again), there also holds
\be{DiffResolv2}
   r'(t) \,=\, - r(t)D \,-\, \int_0^t r(t-s)\gamma(s)\ds\,, \qquad t\geq 0\,.
\ee

The following theorem, the proof of which is given in \ref{App:A} 
now clarifies the existence of a strong solution of \req{GLE}.

\begin{theorem}
\label{Thm:GLE}
Assume that $\gamma\in L^1_{\rm loc}(\R^+_0)$, and let $r$ be the solution of the 
integro-differential equation~\req{DiffResolv}.
Further, let $V_0\in\R^d$ be a centered Gaussian random variable, and
\be{McNg18}
   F \,=\, F_0 \,+\, G\dot{W}
\ee
with a fixed matrix $G\in\R^{d\times d}$, a mean-square bounded and continuous 
$d$-dimensional Gaussian process $F_0$ and
a $d$-dimensional white noise process $\dot{W}$.
Then the generalized Langevin equation~\req{GLE} 
admits a unique strong solution $V$, which satisfies
\be{pathwise}
   V(t) \,=\, r(t)V_0 \,+\, \int_0^t r(t-s)F(s)\ds\,, \qquad t>0\,.
\ee
\end{theorem}

Given the existence of a strong solution of the 
generalized Langevin equation~\req{GLE}, the next question is 
whether this solution is a stationary process
with an associated positive definite covariance matrix
\be{Sigma}
   \Sigma \,=\, \E\bigl(V(t)V(t)^*\bigr)\,.
\ee

For this purpose the following assumption concerning the fluctuating force 
will be made throughout.

\begin{assumption}
\label{Ass:F}
The flcutuating force $F$ of \req{GLE} is a stationary 
Gaussian process\footnote{In fact, $F$ is rather a stationary Gaussian
distribution; see Remark~\ref{Rem:distribution} below.} of the form
\bdm
   F \,=\, F_0 \,+\, G\dot{W}
\edm
with a fixed matrix $G\in\R^{d\times d}$ and
\be{F0phi}
   F_0(t) \,=\, \int_{-\infty}^\infty \varphi(t-s) \dWtilde(s)\,, \qquad t\in\R\,,
\ee
where $W$ and $\Wtilde$ are two-sided $d$-dimensional Brownian motions 
which are either the same or independent of each other.
The density $\varphi$ in \req{F0phi}, which takes values in $\R^{d\times d}$ 
is assumed to satisfy $\varphi\in L^2(\R)$.
\end{assumption}

Take note that the autocorrelation function of $F_0$ of \req{F0phi} is given by
\be{CF0}
   C_{F_0}(t) \,=\, \E\bigl(F_0(s+t)F_0(s)^*\bigr)
   \,=\, \bigl(\varphi * \phitrev^*\bigr)(t)\,, \qquad
   s,t\in\R\,,
\ee
where $*$ denotes convolution and $\phitrev$ is the
time-reversal of $\varphi$, i.e.,
\bdm
   \phitrev(s) \,=\, \varphi(-s)\,, \qquad s\in\R\,.
\edm
It follows that the Fourier transform
\bdm
   \widehat C_{F_0}(\omega) 
   \,=\, \int_{-\infty}^\infty e^{-\rmi \omega t} C_{F_0}(t)\dt\,,
   \qquad \omega\in\R\,,
\edm
of $C_{F_0}$, \emph{aka} the spectral density of $F_0$, satisfies 
\be{CF0hat}
   \widehat C_{F_0} \,=\, \phihat \phihat^* 
\ee
by virtue of the convolution theorem. Since $\varphi$ is assumed to belong
to $L^2(\R)$, this spectral density lies in $L^1(\R)$. 
By the Riemann-Lebesgue lemma, $C_{F_0}$ is therefore a continuous function, 
and hence, $F_0$ is a mean-square bounded and continuous Gaussian process.

And vice versa: If $F_0$ is a mean-square bounded and continuous stationary
Gaussian process which admits a spectral density $\psihat$,
then $\psihat$ belongs to $L^1(\R)$ and Bochner's theorem 
asserts that $\psihat$ can be factorized as $\psihat=\phihat\phihat^*$ for
some $\phihat\in L^2(\R)$; accordingly, $F_0$ can be written in the 
form~\req{F0phi}.

What follows is a rigorous version of the second fluctuation-dissipation
theorem for the generalized Langevin equation~\req{GLE}.
Its proof, and also the one of the subsequent corollary,
are postponed to \ref{App:B}.

\begin{theorem}
\label{Thm:fdt}
Assume that $\gamma\in L^1(\R^+)$ 
and that the fluctuating force satisfies Assumption~\ref{Ass:F} and
is independent of the initial condition $V_0\sim\Normal(0,\Sigma)$,
where $\Sigma\in\R^{d\times d}$ is symmetric and positive definite.
Then the solution of the generalized Langevin equation~\req{GLE}
is a centered Gaussian process; it is a stationary process, if and only if 
\be{GG}
   GG^* \,=\, D\Sigma \,+\, \Sigma D^*
\ee
and 
\be{fdt-relation}
   \varphi*\phitrev^*
   \,=\, \begin{cases}
            \gamma_\Sigma \,-\, \varphi G^* \,-\, G\phitrev^* \,, &
            \text{if\, $\Wtilde = W$}, \\
            \gamma_\Sigma\,, & 
            \text{if\, $W$\! and $\Wtilde$ are independent}\,,
         \end{cases}
\ee
a.e.\ on $\R$, where
\be{gamma-Sigma}
   \gamma_\Sigma(t) \,=\, 
   \begin{cases}
      \gamma(t)\Sigma\,, & t > 0\,, \\
      \Sigma \gamma(-t)^*\,, & t < 0\,, \\
      (\gamma(0)\Sigma \,+\, \Sigma\gamma(0)^*)/2\,, & t=0\,.
   \end{cases}
\ee
In this case the autocorrelation function of $V$\,is given by 
\be{CV}
   C_V(t) \,=\, r(t)\Sigma \qquad \text{for $t\geq 0$}\,,
\ee
where $r$ is the differential resolvent specified in \req{DiffResolv}.
\end{theorem}

Section~\ref{Sec:discussion} below will shed more light on necessary and
sufficient conditions concerning the solvability of the equations~\req{GG} 
and \req{fdt-relation} occurring in Theorem~\ref{Thm:fdt}; 
see Proposition~\ref{Prop:fdt}, in particular.

Concerning the variant~\req{GLE-Kubo} of the generalized Langevin equation
from the introduction
Theorem~\ref{Thm:fdt} has the following implication.

\begin{corollary}
\label{Cor:fdt-Kubo}
Assume that $\gamma\in L^1(\R^+)$, and 
let $V$ be the solution of the generalized Langevin equation~\req{GLE-Kubo} 
with initial condition $V(0)=V_0\sim\Normal(0,\Sigma)$ 
for a mean-square continuous stationary Gaussian fluctuating force term $F$, 
which is independent of $V_0$. Then $V$ is a stationary process, 
if and only if
\be{Cor:fdt-Kubo}
   C_F(t) \,=\, \gamma(t)\Sigma \qquad \text{for $t\geq 0$}\,.
\ee
In this case the autocorrelation function of $V$ is given by
\bdm
   C_V(t) \,=\, r(t)\Sigma \qquad \text{for $t\geq 0$}\,.
\edm
\end{corollary}

Take note that the extension \req{Cor:fdt-Kubo} of 
Kubo's fluctuation-dissipation relation~\req{F-Kubo} to
general covariance matrices $\Sigma$
is well-known. It can already be found in Mori's original 
work~\cite[(3.12)]{Mori65}.

\section{Discussion of the fluctuation-dissipation theorem}
\label{Sec:discussion}
To enhance the interpretation of Theorem~\ref{Thm:fdt}
the subsequent couple of comments and examples may be useful.

\begin{remark}
\label{Rem:fdt-Kubo-ext}
\rm
Consider the setting from the introduction, where \req{GLE-Kubo} describes
the dynamics of a macroparticle, and assume that the fluctuating force has been
chosen according to \req{F-Kubo} such that the solution $V$ of 
\req{GLE-Kubo} is a stationary process,
which satisfies the equipartition theorem, i.e., that its covariance matrix
is given by \req{equi}. 
Note that for \req{F-Kubo} it is necessary that the extension
\be{gamma-ext}
   \gamma(-t) \,=\, \gamma(t)^*\,, \qquad t>0\,,
\ee
of the memory kernel to all $t\in\R$ is a function of positive type
(also called positive definite function in the literature), because the
Fourier transform of a stationary process is a positive semidefinite matrix 
for every frequency $\omega\in\R$.

In this case there may be further stationary solutions of \req{GLE-Kubo} 
which violate the equipartition theorem, because their covariance matrix
$\Sigma$ may not be a multiple of the identity matrix; 
see the following example.
\fin
\end{remark}


\begin{example}
\label{Ex:gamma-Sigma}
\rm
The memory kernel
\bdm
   \gamma(t) \,=\, \begin{cmatrix}
                      e^{-|t|} & 0 \\
                      0 & e^{-2|t|}
                   \end{cmatrix}\,,
\edm
extended as in \req{gamma-ext} to all $t\in\R$, 
is a function of positive type with
\bdm
   \gammahat(\omega) \,=\, 
   \begin{cmatrix}
      2(1+\omega^2)^{-1} & 0 \\[1ex] 0 & 4(4+\omega^2)^{-1}
   \end{cmatrix}\,, \qquad \omega\in\R\,,
\edm
and since $\gammahat\in L^1(\R)$, there exists a fluctuating force $F=F_0$ as
in \req{F0phi}, for which the generalized Langevin equation~\req{GLE-Kubo}
has a stationary solution, which satisfies the equipartition theorem.

For $\eps>0$ and
\bdm
   \Sigma \,=\, \begin{cmatrix} 1 & \eps \\ \eps & 1 \end{cmatrix}
\edm
let $\gamma_\Sigma$ be defined as in \req{gamma-Sigma}. 
Then its Fourier transform is given by
\bdm
   \gammahat_\Sigma(\omega) \,=\, 
   \begin{cmatrix}
      2(1+\omega^2)^{-1} & \eps\,g(\omega) \\[1ex]
      \eps\,\overline{g(\omega)} & 4(4+\omega^2)^{-1}
   \end{cmatrix}\,, \qquad \omega\in\R\,,
\edm
with
\bdm
   g(\omega) \,=\, 
   \frac{1}{1+\rmi\omega} + \frac{1}{2-\rmi\omega}\,.
\edm
Since $|g(\omega)| = O(\omega^{-2})$ as $|\omega|\to\infty$, it follows
that for $\eps$ sufficiently small, $\gammahat_\Sigma$ is diagonally dominant 
for all $\omega\in\R$. This implies that $\gamma_\Sigma$ is a function of
positive type. Furthermore, $\gammahat_\Sigma\in L^1(\R)$, so that the equation
\bdm
   \phihat \phihat^* \,=\, \gammahat_\Sigma
\edm
admits a solution $\phihat\in L^2(\R)$. 
Accordingly, from \req{CF0}, \req{CF0hat}, and Theorem~\ref{Thm:fdt} 
therefore follows that the corresponding 
fluctuating force $F_0$ of \req{F0phi} gives rise to a stationary solution 
of the generalized Langevin equation~\req{GLE-Kubo}, 
but this time with covariance matrix $\Sigma$.
\fin
\end{example}

Example~\ref{Ex:gamma-Sigma} shows that, in general, by varying the initial
condition and the fluctuating force term,
a given generalized Langevin equation
of the form~\req{GLE-Kubo} admits stationary solutions for several different
covariance matrices, but each of the corresponding processes is uniquely 
defined. 

For $D=0$ it follows from the assertion~\req{CV} of Theorem~\ref{Thm:fdt} 
that the autocorrelation function $C_V$ of any stationary solution of \req{GLE} 
satisfies
\be{CV-prime}
   C_V'(0+) \,=\, r'(0)\Sigma \,=\, 0
\ee
by virtue of \req{DiffResolv}; here, $C_V'(0+)$ refers to the right-hand
side derivative of $C_V$ at $t=0$. 
In practice, on the other hand, when the goal is to set up a 
phenomenological model for a stochastic process of interest,
the available observational data may not be consistent with the 
property~\req{CV-prime}. If the data rather indicate that
\bdm
   C_V'(0+) \,=\, -D\Sigma
\edm
for some nonzero $D\in\R^{d\times d}$, then this suggests 
(cf., e.g., \cite{JuSc21}) to use a
generalized Langevin equation of the form \req{GLE} with a corresponding
instantaneous drift term as a potential model for $V$. Take note that 
\bdm
   C_V(-t) \,=\, C_V(t)^*\,,
\edm
and hence, the left-hand side derivative at $t=0$ equals
\bdm
   C_V'(0-) \,=\, \Sigma D^*
\edm
in this case, meaning that the graph of $C_V$ has a kink at $t=0$, 
unless $D\Sigma$ is skew-symmetric.

The generalized Langevin equation~\req{GLE} with general instantaneous drift
coefficient $D$ has been used as a phenomenological model, for example,
by Ceriotti, Bussi, and Parrinello~\cite{CBP10}
and by Ma, Li, and Liu~\cite{MLL16}. For the same situation
McKinley and Nguyen~\cite{McNg18} have investigated a fluctuating force $F$ 
as in Assumption~\ref{Ass:F} with two independent Brownian motions $W$ and
$\Wtilde$, and they proved the following analog of 
Kubo's result\footnote{The papers \cite{CBP10} and \cite{McNg18} consider the 
generalized Langevin equation with infinite time horizon, 
but this is not essential for their arguments.
The generalized Langevin equation with infinite time horizon 
is the topic of Section~\ref{Sec:GLEinf} below.} (which is
a special case of Theorem~\ref{Thm:fdt}): If $C_{F_0}=\gamma/(\beta m)$ 
and $G$ is such that \req{GG} holds true for $\Sigma$ of \req{equi},
and if the solution $V$ of \req{GLE} is a stationary 
process, then $V$ satisfies the equipartition theorem. 
It must be emphasized, though, that this choice
of the fluctuating force requires that the extension~\req{gamma-ext}
of the memory kernel to all $t\in\R$ 
be a function of positive type (as pointed out
in Remark~\ref{Rem:fdt-Kubo-ext}).
Theorem~\ref{Thm:fdt} relaxes this requirement by admitting a larger class
of fluctuating force terms, and 
as will be seen in Example~\ref{Ex1} below,
stationary solutions of \req{GLE}, which satisfy the equipartition theorem,
may also exist for certain memory kernels for which the extension to all of 
$\R$ fails to be a function of positive type. 
The reason for this is the following result.
 
\begin{proposition}
\label{Prop:fdt}
Let $\Sigma\in\R^{d\times d}$ be symmetric positive definite, and
let the memory kernel $\gamma$ of \req{GLE} belong to $L^1(\R^+)$.
Assume further that the Fourier transform of $\gamma_\Sigma$ of 
\req{gamma-Sigma} belongs to $L^1(\R)$. 
Then the following two statements are equivalent:
\renewcommand{\theenumi}{(\roman{enumi})}
\begin{itemize}
\item[(i)]
There exist matrices $G\in\R^{d\times d}$ which satisfy \req{GG}, and for
any such $G$ the fluctuation-dissipation relation
\be{Prop:fdt}
   \varphi*\phitrev^* \,+\, \varphi G^* \,+\, G\phitrev^* \,=\, \gamma_\Sigma
   \qquad \text{a.e.\ on $\R$}
\ee
has a 
unique solution $\varphi\in L^2(\R)$ with values in $\R^{d\times d}$.
\item[(ii)]
The Fourier transform of $\gamma_\Sigma$ satisfies
\be{SylvesterI}
   \gammahat_\Sigma(\omega) \,+\, D\Sigma \,+\, \Sigma D^* \,\Loewnergeq\, 0
   \qquad \text{for all $\omega\in\R$}\,,
\ee
meaning that the left-hand side of \req{SylvesterI} is a positive semidefinite
matrix for every $\omega\in\R$.
\end{itemize}
\end{proposition}

\begin{proof}
If condition~(i) holds true, then it follows from \req{Prop:fdt}
and the convolution theorem that
\bdm
   \phihat(\omega)\phihat(\omega)^*
         \,+\, \phihat(\omega)G^* \,+\, G\,\phihat(\omega)^* 
   \,=\, \gammahat_\Sigma(\omega) \qquad \text{for a.e.\ $\omega\in\R$}\,.
\edm
This further implies that
\bdm
   (\phihat(\omega)+G)(\phihat(\omega)+G)^*
   \,=\, \gammahat_\Sigma(\omega) \,+\, GG^*
   \,=\, \gammahat_\Sigma(\omega) \,+\, D\Sigma \,+\, \Sigma D^*
\edm
by virtue of \req{GG}, and since the left-hand side of this equation is
positive semidefinite, \req{SylvesterI} follows. Strictly speaking,
the argument only applies almost everywhere,
but since the memory kernel is taken to belong to $L^1(\R^+)$, 
the left-hand side of \req{SylvesterI}
is continuous in $\omega$. Therefore \req{SylvesterI} is true for every
$\omega\in\R$.

Consider now the case that condition~(ii) is satisfied.
Since $\gamma_\Sigma\in L^1(\R)$, $\gammahat_\Sigma$ is a continuous function 
and $\gammahat_\Sigma(\omega)$
tends to zero as $|\omega|\to\infty$ by the Riemann-Lebesgue lemma. 
From \req{SylvesterI} therefore follows
that $D\Sigma + \Sigma D^*$ is positive semidefinite, and hence, \req{GG} has
solutions $G\in\R^{d\times d}$. Let $G$ be any solution of this equation,
and $G^*=PG_0$ be the polar decomposition of $G^*$, with a symmetric positive
definite matrix $G_0$ and an orthogonal matrix $P$. Then 
\bdm
   (D\Sigma+\Sigma D^*)^{1/2} \,=\, G_0\,,
\edm
and
an estimate of Vainikko (see~\cite[p.~91]{VaVe86}) asserts 
that the spectral norm of 
\begin{align*}
   \phihat(\omega) 
   &\,:=\, \bigl(D\Sigma+\Sigma D^*+\gammahat_\Sigma(\omega)\bigr)^{1/2}P^*
           \,-\, G \\
   &\,\phantom{:}=\, \Bigl(
             \bigl(D\Sigma+\Sigma D^*+\gammahat_\Sigma(\omega)\bigr)^{1/2}
                   \,-\,(D\Sigma+\Sigma D^*)^{1/2}
          \Bigr)P^*
\end{align*}
satisfies
\be{Vainikko}
   \norm{\phihat(\omega)}_2 
   \,\leq\, \frac{4}{\pi}\, \norm{\gammahat_\Sigma(\omega)}_2^{1/2}\,.
\ee
Moreover, for every $\omega\in\R$ there holds
\begin{align}
\nonumber
   \phihat(\omega)\phihat(\omega)^*
   &\,=\, \bigl(\phihat(\omega) + G\bigr)\bigl(\phihat(\omega)^*+G^*\bigr)
          \,-\, \phihat(\omega)G^* \,-\, G\,\phihat(\omega)^*
          \,-\, (D\Sigma+\Sigma D^*)\\[1ex]
\label{eq:FT:fdt-relation}
   &\,=\, \gammahat_\Sigma(\omega) \,-\, \phihat(\omega)G^*
          \,-\, G\,\phihat(\omega)^*\,.
\end{align}

From \req{Vainikko} follows that $\phihat\in L^2(\R)$, 
hence $\phihat$ is the Fourier transform of some function $\varphi\in L^2(\R)$. 
Furthermore, since $\Sigma$ and the memory kernel $\gamma$ are assumed to be 
real-valued, 
there holds $\gammahat_\Sigma(-\omega)=\overline{\gammahat_\Sigma(\omega)}$, 
and therefore the
same property is true for $\phihat$. Accordingly, the matrix entries of
$\varphi(t)$ are real-valued for every $t\in\R$.
Finally, as the left-hand side of \req{FT:fdt-relation} is the Fourier 
transform of the convolution $\varphi*\phitrev^*$, it follows from an 
inverse Fourier transform of \req{FT:fdt-relation}
that $\varphi$ solves \req{Prop:fdt}.
\end{proof}

\begin{remark}
\label{Rem:L1}
\rm
Proposition~\ref{Prop:fdt} requires on top of the assumptions of 
Theorem~\ref{Thm:fdt} that $\gammahat_\Sigma\in L^1(\R)$. 
On the other hand, if $\Sigma$ is such that $\gamma(0)\Sigma$
is not a symmetric matrix, then $\gamma_\Sigma$ fails to be continuous at
$t=0$, and hence, $\gammahat_\Sigma$ cannot belong to $L^1(\R)$.
Nonetheless, Proposition~\ref{Prop:fdt} remains valid
-- independent of whether $\gamma(0)\Sigma$ is symmetric or not --
when the memory kernel is sufficiently smooth, for example, when
$\gamma$ is absolutely integrable and known to belong to the Sobolev space 
$H^1(\R^+)$ of all square integrable absolutely continuous functions $\gamma$, 
which can be written as
\bdm
   \gamma(t) \,=\, \gamma(0) \,+\, \int_0^t \gamma'(s)\ds \qquad 
\edm
for some $\gamma'\in L^2(\R^+)$ and all $t>0$.
In fact, it is only the second assertion (ii)$\Rightarrow$(i) 
which is problematic in this case because the extra assumption 
$\gammahat_\Sigma\in L^1(\R)$ is not required for the other statement.

In the sequel it is briefly sketched how the argument in the proof of
Proposition~\ref{Prop:fdt} can be modified in this case;
to simplify matters, it will be assumed
that $D\Sigma+\Sigma D^*$ is positive definite and that $G$ of \req{GG}
is its square root.

Let
\bdm
   \psi_\Sigma(t) \,=\, 
   \begin{cases}
      \gamma_\Sigma(t) \,-\, \gamma(0)\Sigma e^{-t}\,, & t>0\,, \\
      0\,, & t=0\,, \\
      \gamma_\Sigma(t) \,-\, \Sigma\gamma(0)^*e^t\,, & t<0\,,
   \end{cases}
\edm
and take note that $\psi_\Sigma\in L^1(\R)\cap H^1(\R)$. Then one computes
\be{Rem:L1-1}
   \gammahat_\Sigma(\omega)
   \,=\, \frac{1}{\rmi\omega}\,\widehat{\psi'_\Sigma}(\omega)
         \,+\, \frac{1}{1+\omega^2}
               \bigl(\gamma(0)\Sigma+\Sigma\gamma(0)^*\bigr)
         \,+\, A(\omega)\,,
\ee
where the Hermitian matrix $A(\omega)$ is given by
\be{Rem:L1-2}
   A(\omega) \,=\, \rmi\bigl(\Sigma\gamma(0)^*-\gamma(0)\Sigma\bigr)
                   \frac{\omega}{1+\omega^2}\,, \qquad
   \omega\in\R\,.
\ee
For $\omega\in\R$ denote by
\bdm
   X(\omega) \,=\, \int_0^\infty e^{-tG}A(\omega)e^{-tG}\dt
\edm
the solution of the Lyapunov equation
\bdm
   GX(\omega) \,+\, X(\omega)G \,=\, A(\omega)\,.
\edm
As functions of $\omega$, $A$ and $X$ belong to $L^2(\R)$.

Introducing
\bdm
   \phihat_0(\omega) 
   \,=\, \bigl(D\Sigma+\Sigma D^*+\gammahat_\Sigma(\omega)\bigr)^{1/2} 
         \,-\, \bigl(G+X(\omega)\bigr)\,,
\edm
it follows that
\bdm
   \norm{\phihat_0(\omega)}_2
   \,\leq\, \frac{4}{\pi}\,
        \norm{\gammahat_\Sigma(\omega) \,-\, A(\omega) \,-\, X(\omega)^2}_2^{1/2}
\edm
as in the proof of Proposition~\ref{Prop:fdt}, and hence
$\phihat_0\in L^2(\R)$, because $\gammahat_\Sigma-A$ belongs
to $L^1(\R)$ by virtue of \req{Rem:L1-1}, and so does
$X^2$. Now it is easy to see that $\phihat=\phihat_0+X\in L^2(\R)$ satisfies
\req{FT:fdt-relation}, in which case its inverse Fourier transform $\varphi$
solves the fluctuation-dissipation relation~\req{Prop:fdt}.
\fin
\end{remark}

\begin{remark}
\label{Rem:phisymmetry}
\rm
The proof of the assertion (ii)$\Rightarrow$(i) in Proposition~\ref{Prop:fdt}
can be used for a numerical construction of $\varphi$ of \req{fdt-relation}
and the associated component $F_0$ of the fluctuating force in \req{GLE},
when $\Wtilde=W$ in Assumption~\ref{Ass:F}.
This construction simplifies somewhat when $G$ is chosen to be 
$(D\Sigma+\Sigma D^*)^{1/2}$: then $P=I$ in the polar decomposition of $G^*=G$, 
$\phihat(\omega)$ is self-adjoint for every $\omega\in\R$, and
\be{phisymmetry}
   \varphi(-t) \,=\, \varphi(t)^* \qquad \text{for every $t\in\R$}\,.
\ee
Accordingly, $\phitrev^*=\varphi$ in \req{fdt-relation}.

The identity~\req{phisymmetry} is always true in the scalar case, 
where $\phihat$ is a real-valued scalar function.
\fin
\end{remark}

As has been announced above, the following example demonstrates that there are
memory kernels, for which an instantaneous friction term in \req{GLE} 
is necessary to find a stationary solution of the 
generalized Langevin equation which satisfies the equipartition theorem.

\begin{example}
\label{Ex1}
\rm
The function
\be{Ex1}
   \gamma(t) \,=\, 4\thsp|t|\,e^{-|t|}\,, \qquad t\in\R\,,
\ee
fails to be a function of positive type, because $\gamma(0)$ is not its global
maximum. Accordingly, there cannot be a stationary random process $F$ or $F_0$
satisfying~\req{F-Kubo}. However,
\bdm
   \gammahat(\omega) \,+\, 1
   \,=\, 8\frac{1-\omega^2}{(1+\omega^2)^2} \,+\, 1
   \,=\, \Bigl(\frac{\omega^2-3}{\omega^2+1}\Bigr)^2 \,\geq\, 0
\edm
for every $\omega\in\R$, with $\gammahat\in L^1(\R)$. 
Therefore assumption~\req{SylvesterI} of Proposition~\ref{Prop:fdt} 
is satisfied for $\Sigma=1$ and $D=1/2$.
Choosing $G=1$ in this example, the unique solution $\varphi\in L^2(\R)$ 
of \req{Prop:fdt} is given by
\be{Ex1:phi}
   \varphi(t) \,=\, -2\,e^{-|t|}\,, \qquad t\in\R\,,
\ee
because, cf.~Remark~\ref{Rem:phisymmetry},
\bdm
   (\varphi*\varphi)(t) \,+\, 2\, \varphi(t) \,=\, 
   4\,(1+|t|)e^{-|t|} \,-\, 4\, e^{-|t|} \,=\, \gamma(t)\,.
\edm
Therefore, the generalized Langevin equation~\req{GLE} 
with memory kernel~\req{Ex1} and instantaneous friction coefficient $D=1/2$ 
has a stationary solution $V$ with variance $\sigma^2=1$, 
if the fluctuating force $F$ is defined as in 
Assumption~\ref{Ass:F} with $\Wtilde=W$ and $\varphi$ of \req{Ex1:phi}, 
and if the initial condition $V_0$ is (independently) normally distributed. 
When rescaling $V_0$ and the fluctuating force by the additional factor 
$1/(\beta m)$, then \req{GLE} has a stationary solution which satisfies 
the equipartition theorem.
\fin
\end{example} 

\begin{remark}
\label{Rem:distribution}
\rm
The fluctuating force $F$ of Assumption~\ref{Ass:F} is a stationary
random distribution in the sense of It\^{o}~\cite{Ito53}, compare 
also \cite{KoSi07,Yagl86} for an exposition of the corresponding theory:
For arbitrary test functions $f,g\in C_0^\infty(\R)$
with values in $\R^d$, and for a given path of $W$ and $\Wtilde$ 
the action of $F$ on $f$ is defined as
\begin{align*}
   \scalp{F,f} 
   &\,=\, \int_{-\infty}^\infty f(t)^*F_0(t)\dt 
          \,+\, \int_{-\infty}^\infty f(t)^* G\dW(t)\\[1ex]
   &\,=\, \int_{-\infty}^\infty\int_{-\infty}^\infty f(t)^* \varphi(t-s)\dWtilde(s)\dt
          \,+\, \int_{-\infty}^\infty f(t)^* G\dW(t)\,.          
\end{align*}
Accordingly, the covariance structure of $F$ is determined by
\begin{align*}
   &\E\bigl(\scalp{F,f}\scalp{F,g}\bigr)\\[1ex]
   &\!=\, \int_{-\infty}^\infty \int_{-\infty}^\infty \int_{-\infty}^\infty
             f(t)^* \varphi(t-s) \varphi(\tau-s)^* g(\tau)\ds\dtau\dt\\[1ex]
   &\phantom{\!=\ }
          +\, \kappa\,\Bigl(
                      \int_{-\infty}^\infty\int_{-\infty}^\infty
                         f(t)^*\varphi(t-\tau)G^*g(\tau)\dtau\dt 
                      \,+\,
                      \int_{-\infty}^\infty\int_{-\infty}^\infty
                         f(t)^*G\varphi(\tau-t)^*g(\tau)\dtau\dt\Bigr)\\[1ex] 
   &\phantom{\!=\ }
          +\, \int_{-\infty}^\infty f(t)^*GG^*g(t)\dt\,,
\end{align*}
where $\kappa=0$ if $W$ and $\Wtilde$ are independent, and $\kappa=1$ if 
they are the same. It follows that
\bdm
   \E\bigl(\scalp{F,f}\scalp{F,g}\bigr)
   \,=\, \int_{-\infty}^\infty\int_{-\infty}^\infty 
            f(t)^*C_F(t-\tau)g(\tau)\dtau\dt\,,
\edm
and the associated autocorrelation distribution of $F$ is given by
\bdm
   C_F \,=\, \varphi*\phitrev^* 
             \,+\, \kappa\bigl(\varphi G^* + G\phitrev^*\bigr)
             \,+\, GG^*\,\delta\,,
\edm
with $\delta$ denoting the delta distribution. 

The fluctuation-dissipation theorem~\ref{Thm:fdt} therefore asserts that if
(and only if)
\bdm
   C_F \,=\, \gamma\Sigma \,+\, (D\Sigma+\Sigma D^*)\,\delta\,,
\edm
then the generalized Langevin equation~\req{GLE} has a stationary solution.
This form of the fluctuation-dissipation relation can be found
in \cite{MLL16} for the particular matrix $\Sigma$ of \req{equi}. 
Take note, though, that in \cite{MLL16} the delta distribution is 
weighted with $2D/(\beta m)$, rather than 
$D\Sigma+\Sigma D^*=(D+D^*)/(\beta m)$; 
apparently, the corresponding authors presume their coefficient matrix $D$ 
to be symmetric.
\fin
\end{remark}

\section{The generalized Langevin equation with infinite time horizon}
\label{Sec:GLEinf}
Consider now the equation
\be{GLEinf}
   \dot{V}_\infty(t) \,=\, -DV_\infty(t) 
      \,-\, \int_{-\infty}^t \gamma(t-s)V_\infty(s)\ds \,+\, F(t)\,,\qquad
   t\in\R\,,
\ee
where the integration starts at minus infinity and the fluctuating force
is given by
\be{McNg18-inf}
   F \,=\, F_0 \,+\, G\dot{W}
\ee
with $G\in\R^{d\times d}$ and $F_0$ a mean-square bounded and continuous
$d$-dimensional Gaussian process; here, $W$ is a two-sided $d$-dimensional
Brownian motion.

\begin{theorem}
\label{Thm:GLEinf}
Let $\gamma\in L^1(\R^+)$ and $F$ be given by \req{McNg18-inf}.
 Assume further that the solution $r$ of the 
integro-differential equation~\req{DiffResolv} belongs to $L^1(\R^+)$. Then
\be{Vinfty}
   V_\infty(t) \,=\, \int_{-\infty}^t r(t-s)F(s)\ds
\ee
is the unique bounded strong solution of \req{GLEinf}. 
Moreover, if $V$ is the corresponding solution of \req{GLE} with the same 
fluctuating force term restricted to $t>0$, 
and any centered Gaussian random variable $V_0\in\R^d$, then
\bdm
   \E\bigl(\norm{V_\infty(t)-V(t)}_2^2\bigr) \,\to\, 0\,, \qquad t\to\infty\,.
\edm
\end{theorem}

The proof of this result, which is similar to the one of Theorem~\ref{Thm:GLE},
can also be found in \ref{App:A}. The major difference is the fact
that the differential resolvent $r$ now needs to belong to $L^1(\R^+)$
in order to bound the contribution of the fluctuating force from
the infinite time horizon.

\begin{remark}
\label{Rem:PW}
\rm
Concerning this additional requirement,
the Paley-Wiener theory gives a necessary and sufficient condition for it
to be true: 
Namely, there holds $r\in L^1(\R^+)$, if and only if the $d\times d$ matrix
\bdm
\begin{array}{c}
   A(\zeta) \,=\, \zeta I \,+\, D \,+\, \L\gamma(\zeta)
   \\[1.5ex]
   \text{is invertible for every $\zeta\in\C$ with $\Real\zeta\geq 0$}\,,
\end{array}
\edm
compare \cite[Theorem~3.3.5]{GLS90},
where 
\bdm
   \L\gamma(\zeta) \,=\, \int_0^\infty e^{-\zeta t}\gamma(t)\dt
\edm
denotes the Laplace transform of $\gamma$.

If $\Sigma\in\R^{d\times d}$ is any symmetric and positive definite matrix,
the above criterion is fulfilled, if and only if $A(\zeta)\Sigma x\neq 0$ 
for every $\zeta\in\C$ with $\Real\zeta\geq 0$ and every 
$x\in\C^n\setminus\{0\}$.
Since $\gamma$ belongs to $L^1(\R^+)$, the function 
$\zeta\mapsto x^*A(\zeta)\Sigma x$ for fixed $x\in\C^n\setminus\{0\}$ is 
analytic in the complex right half plane with continuous boundary values on 
the imaginary axis. Its real part being a harmonic function, 
it follows from the maximum principle that $x^*A(\zeta)\Sigma x\neq 0$ 
for all $\zeta$ in the closed right-half plane, if
\bdm
   2\,\Real \bigl(x^*A(\rmi\omega)\Sigma x\bigr)
   \,=\, x^*\bigl(
               A(\rmi\omega)\Sigma\,+\, \Sigma A(\rmi\omega)^*
            \bigr)x \,>\, 0 
\edm
for every $\omega\in\R$.
Accordingly, a sufficient condition for $r$ to belong to $L^1(\R^+)$ is that
\bdm
   A(\rmi\omega)\Sigma \,+\, \Sigma A(\rmi\omega)^*
   \,\Loewnergr\, 0 \qquad \text{for all $\omega\in\R$}\,,
\edm
i.e., that the corresponding matrices are positive definite.
Recurring to the function $\gamma_\Sigma$ defined in \req{gamma-Sigma},
the above criterion is equivalent to
\be{PW}
   D\Sigma \,+\, \Sigma D^* \,+\, \gammahat_\Sigma(\omega) \,\Loewnergr\, 0
   \qquad \text{for all $\omega\in\R$}\,.
\ee
In summary: If \req{PW} holds true for any symmetric and positive definite
matrix $\Sigma$, then the differential resolvent $r$ belongs to $L^1(\R^+)$.
\fin
\end{remark}

The following theorem, the proof of which is given in \ref{App:C},
is the main result of this section.
It shows that \emph{every} bounded solution of \req{GLEinf}
is a stationary one, provided the fluctuating force satisfies 
Assumption~\ref{Ass:F}. 

\begin{theorem}
\label{Thm:fdtinf}
Assume that $\gamma\in L^1(\R^+)$, and that the differential resolvent
$r$ of \req{DiffResolv} also belongs to $L^1(\R^+)$.
If $F$ satisfies Assumption~\ref{Ass:F} then the unique bounded solution
of \req{GLEinf} is a stationary Gaussian process with autocorrelation
function
\be{Thm:fdtinf}
   C_{V_\infty}(t) 
   \,=\, \int_0^\infty r(t+\tau)GG^*r(\tau)^*\dtau
         \,+\, \int_0^\infty\int_0^\infty
                  r(\tau)\chi(t+\tau'-\tau)r(\tau')^*\dtau'\dtau
\ee
for $t\geq 0$, where
\be{chiinf}
   \chi \,=\, \begin{cases}
                 \varphi G^* \,+\, G\phitrev^* \,+\, \varphi*\phitrev^* \,, &
                 \text{if \ $\Wtilde = W$}\,, \\
                 \varphi*\phitrev^*\,, & 
                 \text{if \ $W$ and $\Wtilde$ are independent}\,.
              \end{cases}
\ee
Moreover, if $\Sigma\in\R^{d\times d}$ is symmetric positive definite
and $\gamma_\Sigma$ is given by \req{gamma-Sigma}, then the 
autocorrelation function of $V_\infty$ can alternatively be written as
\be{Thm:fdtinf2}
\begin{aligned}
   C_{V_\infty}(t) 
   &\,=\, r(t)\Sigma
      \,+\, \int_0^\infty r(t+\tau)(GG^*-D\Sigma-\Sigma D^*)r(\tau)^*\dtau\\[1ex]
   &\phantom{\,=\ }
      +\, \int_0^\infty\int_0^\infty
               r(\tau)(\chi-\gamma_\Sigma)(t+\tau'-\tau)r(\tau')^*\dtau'\dtau\,.
\end{aligned}
\ee
\end{theorem}

It follows from Theorems~\ref{Thm:GLEinf} and \ref{Thm:fdtinf} 
under the given assumptions that
the unique solution of \emph{every} initial value problem~\req{GLE} approaches
some stationary Gaussian process as $t\to+\infty$, 
whenever the fluctuating force $F$ satisfies Assumption~\ref{Ass:F}. 
In a physical context
some of them may be consistent with the equipartition theorem, others may not;
see the following corollaries and remarks.

\begin{corollary}
\label{Cor:fdt-relinf}
Let $\Sigma\in\R^{d\times d}$ be symmetric positive definite 
and $\gamma\in L^1(\R^+)$.
Assume further that the Fourier transform of $\gamma_\Sigma$ of 
\req{gamma-Sigma} belongs to $L^1(\R)$, and that
\be{SylvesterII}
   \gammahat_\Sigma(\omega) \,+\, D\Sigma \,+\, \Sigma D^* \,\Loewnergr\, 0
   \qquad \text{for all $\omega\in\R$}\,.
\ee
Then there exists $G\in\R^{d\times d}$ satisfying
\be{GGII}
   D\Sigma \,+\, \Sigma D^* \,=\, GG^*\,,
\ee
and for $\Wtilde=W$ and any such $G$
the fluctuation-dissipation relation~\req{fdt-relation} 
has a unique solution $\varphi\in L^2(\R)$ with values in $\R^{d\times d}$.
With this choice of $G$ and $\varphi$ and the corresponding fluctuating force 
$F$ of Assumption~\ref{Ass:F}, the bounded solution $V_\infty$
of the generalized Langevin equation~\req{GLEinf} with infinite time horizon
is a stationary Gaussian process with autocorrelation function
\bdm
   C_{V_\infty}(t) \,=\, r(t)\Sigma\,, \qquad t\geq 0\,,
\edm
where $r$ is the associated differential resolvent~\req{DiffResolv}.
\end{corollary}

\begin{proof}
Due to assumption~\req{SylvesterII} and Remark~\ref{Rem:PW}
the differential resolvent $r$ belongs to $L^1(\R^+)$.
Furthermore, it follows from Proposition~\ref{Prop:fdt} that there exists
a solution $G\in\R^{d\times d}$ of \req{GGII} and an associated density 
$\varphi\in L^2(\R)$ with values in $\R^{d\times d}$ satisfying the 
fluctuation-dissipation relation~\req{fdt-relation}.
By Theorem~\ref{Thm:fdtinf}, the bounded solution $V_\infty$ of \req{GLE} 
with the associated fluctuating force $F$ is a stationary Gaussian process, 
and from \req{Thm:fdtinf2} it follows that its 
autocorrelation function is given by $r(t)\Sigma$, because $\chi$ of 
\req{chiinf} coincides with $\gamma_\Sigma$ according to \req{fdt-relation}.
\end{proof}

\begin{remark}
\label{Rem:McNg18}
\rm
If $\gamma_\Sigma$ itself is a function of positive type, 
and if $D\Sigma+\Sigma D^*$ is positive definite, 
then one can define $F$ in Assumption~\ref{Ass:F}
with two independent Brownian motions $W$ and $\Wtilde$ to obtain the same
stationary solution as in Corollary~\ref{Cor:fdt-relinf}. 
For $\Sigma$ of \req{equi} this result, which is
compatible with the equipartition theorem, was obtained in \cite{McNg18}.
\fin
\end{remark}

\begin{remark}
\label{Rem:Riedle}
\rm
Riedle~\cite{Ried03} proved a special case of Theorem~\ref{Thm:fdtinf},
in which $D=0$ in \req{GLEinf} and the fluctuating force $F$ 
is pure white noise, i.e., $\varphi=0$ in Assumption~\ref{Ass:F}. 
He showed that the solution
$V_\infty$ constructed in Theorem~\ref{Thm:GLEinf} is a stationary Gaussian
process, and its autocorrelation function is given by
\be{Rem:Riedle}
   C_{V_\infty}(t) \,=\, \int_0^\infty r(t+\tau)r(\tau)^*\dtau\,, \qquad t\geq 0\,;
\ee
compare \req{Thm:fdtinf} of Theorem~\ref{Thm:fdtinf}. 
\fin
\end{remark}

In general the covariance matrix $C_{V_\infty}(0)$ of the stationary solution
of \req{GLEinf} is difficult to analyze; 
this is apparent from \req{Rem:Riedle}, for example. 
However, in the scalar case it is obvious that with an appropriate rescaling of 
any of the stationary fluctuating forces of Assumption~\ref{Ass:F} 
the corresponding solution $V_\infty$ of \req{GLEinf} can have any prescribed
variance. Nonetheless, the resulting processes
(e.g., the ones of Remark~\ref{Rem:McNg18} and Remark~\ref{Rem:Riedle})
will be different, because their autocorrelation functions are.

Concerning the variety of possible stationary solutions of \req{GLEinf}
the following can be shown.

\begin{corollary}
\label{Cor:fdtinf-D0}
Consider the generalized Langevin equation~\req{GLEinf} with $D=0$ and
$\gamma\in L^1(\R^+)$, and assume that the differential resolvent $r$ 
of \req{DiffResolv} also belongs to $L^1(\R^+)$.
Let $\psi\in H^3(\R)$ be a function of positive type with three (weak)
derivatives in $L^2(\R)$.
Then there exists a stationary fluctuating force $F=F_0$ given by
\req{F0phi} for some appropriate $\varphi\in L^2(\R)$, 
such that the autocorrelation function 
of the stationary solution $V_\infty$ of \req{GLEinf}
coincides with $\psi$. 
\end{corollary}

Since the proof of this corollary requires some intermediate result from
the proof of Theorem~\ref{Thm:fdt}, this proof is postponed to \ref{App:C}.

Corollary~\ref{Cor:fdtinf-D0} shows that a huge class of 
different stationary Gaussian processes $V_\infty$ can be obtained
as solutions of \req{GLEinf} by varying
the fluctuating force term. Infinitely many of them will share the same
covariance matrix. In other words, as far as the setting from the introduction
for the model of the dynamics of a macroparticle is concerned,
there are infinitely many different ways of choosing the fluctuating force
in order to satisfy the equipartition theorem asymptotically.

As mentioned in the introduction it is of certain interest to also consider
the cross-correlations between the fluctuating force and the 
stationary solution of \req{GLEinf}.
The outcome, which is also proved in \ref{App:C}, is as follows; 
see also \cite{SKTL10,Hank21,JuSc21}.

\begin{proposition}
\label{Prop:CFV}
Under the assumptions of Theorem~\ref{Thm:fdtinf}, let
$F$ be the fluctuating force and $V_\infty$ be the stationary solution 
of the generalized Langevin equation~\req{GLEinf} with infinite time horizon. 
Then for $s\in\R$ and $t\neq 0$ there holds
\be{Prop:CFV}
   \E\bigl(V_\infty(s+t)F(s)^*\bigr)
   \,=\, \int_0^\infty r(\tau)\chi(t-\tau)\dtau
         \,+\,
         \begin{cases}
            r(t)GG^*\,, & t>0\,,\\
            0\,, & t<0\,,
         \end{cases}
\ee
with $\chi$ of \req{chiinf}.
\end{proposition}

In contrast to the formulation~\req{GLE} of the generalized Langevin
equation as an initial value problem, 
the correlations between velocities and forces 
for the equation with infinite time horizon only depend on the
time lag. Accordingly, the joint process $(V_\infty,F)$ is a stationary one.

Concerning the original setting of Kubo, where 
$D=G=0$, and $\chi$ of \req{chiinf} satisfies
\bdm
   \chi \,=\, \varphi*\phitrev^* \,=\, C_F \,=\, \gamma_\Sigma\,,
\edm
Proposition~\ref{Prop:CFV} shows that for the infinite time horizon limiting
system the cross-correlation between $F(t)$ for $t>0$ and $V(0)$ is given by
\bdmal
   \E\bigl(F(t)V(0)^*\bigr) 
   &\,=\, \E\bigl(V(-t)F(0)^*\bigr)^*
    \,=\, \left(\int_0^\infty r(\tau)\gamma_\Sigma(-t-\tau)\dtau\right)^* \\[1ex]
   &\,=\, \int_0^\infty \gamma(t+\tau)\Sigma r(\tau)^*\dtau\,,
\edmal 
which is nonzero in general. Recall that Kubo's analysis of the generalized
Langevin equation~\req{GLE-Kubo}, on the other hand, assumes that $V(0)$
and the fluctuating force at future times are uncorrelated.

\setcounter{section}{0}
\renewcommand{\thesection}{Appendix~\Alph{section}}
\renewcommand{\theequation}{\Alph{section}.\arabic{equation}}
\section{Proofs of Theorems~\ref{Thm:GLE} and \ref{Thm:GLEinf}}
\label{App:A}
This first appendix contains the proofs of the two existence theorems
for the solution of the generalized Langevin equation with and without
infinite time horizon.

\begin{proofof}{Theorem~\ref{Thm:GLE}}
\begin{subequations}
\label{eq:V12}
For the generalized Langevin equation~\req{GLE}, formulated as an initial
value problem starting at time $t=0$, let
\be{V1}
   V_1(t) \,=\, r(t) V_0 \,+\, \int_0^t r(t-\tau)F_0(\tau)\dtau
\ee
and
\be{V2}
   V_2(t) \,=\, \int_0^t r(t-\tau)\thsp G\,\dW(\tau) 
\ee
\end{subequations}
for all $t\geq 0$ and a given path of the fluctuating force,
in accordance with the decomposition~\req{McNg18}. 
Since $r'$ is continuous, $V_1$ has a mean-square continuous derivative 
\be{V1prime}
   V_1'(t) \,=\, 
   r'(t) V_0 \,+\, F_0(t) \,+\, \int_0^t r'(t-\tau)F_0(\tau)\dtau\,,
\ee
while $V_2$ can be rewritten as
\bdm
   V_2(t) \,=\, GW(t) \,+\, \int_0^t r'(t-\tau)GW(\tau)\dtau\,,
\edm
which reveals that the smoothness properties of $V_2$ are the same as those of 
the Brownian motion. Furthermore, $V_1(0)=V_0$ and $V_2(0)=0$.

The plan is to prove that
\begin{subequations}
\label{eq:V12-GLE}
\be{V1-GLE}
   V_1'(t) \,=\, -DV_1(t) \,-\, \int_0^t \gamma(t-s)V_1(s)\ds \,+\, F_0(t)\,,
\ee
and that $V_2$ is a strong solution of 
\be{V2-GLE}
   \dot{V}_2(t) \,=\, -DV_2(t) \,-\, \int_0^t \gamma(t-s)V_2(s)\ds 
                      \,+\, G\,\dot{W}(t)\,.
\ee
\end{subequations}
Since $V$ of \req{pathwise} is given by $V=V_1+V_2$, this implies 
that $V$ is the strong solution of \req{GLE}.

Consider \req{V1-GLE} first. 
Inserting \req{DiffResolv} into \req{V1prime} gives
\bdmal
   V_1'(t) 
   &\,=\, -Dr(t)V_0 \,-\, \int_0^t \gamma(t-s)r(s)V_0 \ds \,+\, F_0(t) \\[1ex]
   &\phantom{\,=\ }
   \,-\, \int_0^tDr(t-s)F_0(s)\ds 
   \,-\, \int_0^t\int_0^{t-\tau} \gamma(t-\tau-s)r(s)\ds \,F_0(\tau)\dtau\\[1ex]
   &\,=\, -D \left(r(t)V_0 \,+\, \int_0^t r(t-s)F_0(s)\ds\right) \,+\, F_0(t)
          \\[1ex]
   &\phantom{\,=\ }
   \,-\, \int_0^t \gamma(t-s)r(s)V_0\ds 
   \,-\, \int_0^t\gamma(t-s)\int_0^s r(s-\tau)F_0(\tau)\dtau\ds\,.
\edmal
Therefore \req{V1-GLE} follows from the definition~\req{V1} of $V_1$.

In order to show that $V_2$ is a strong solution of \req{V2-GLE} 
one has to show that
\be{V2-strongsol}
   V_2(t) \,=\, -\int_0^t \left(
                   DV_2(\tau) \,+\, \int_0^\tau \gamma(\tau-s)V_2(s)\ds
                 \right)\!\dtau \,+\, GW(t)\,.
\ee
Using \req{V2} and \req{DiffResolv} the integrand
on the right-hand side of \req{V2-strongsol} can be rewritten as
\begin{align*}
   &DV_2(\tau) \,+\, \int_0^\tau \gamma(\tau-s)V_2(s)\ds\\[1ex]
   &\quad 
    \,=\, D \int_0^\tau r(\tau-\tau')G\dW(\tau')
          \,+\, \int_0^\tau 
                   \gamma(\tau-s)\int_0^s r(s-\tau')G\dW(\tau')\ds\\[1ex]
   &\quad
    \,=\, D \int_0^\tau r(\tau-\tau')G\dW(\tau')
          \,+\, \int_0^\tau\int_0^{\tau-\tau'} 
                  \gamma(\tau-\tau'-s') r(s')\ds'\, G\dW(\tau')\\[1ex]
   &\quad
    \,=\, \int_0^\tau \left(
             D r(\tau-\tau')\,+\,\int_0^{\tau-\tau'}\gamma(\tau-\tau'-s')r(s')\ds'
          \right)G\dW(\tau') \\[1ex] 
   &\quad
    \,=\, - \int_0^\tau r'(\tau-\tau') G\dW(\tau') \,,
\end{align*}
and hence,
\begin{align*}
   &-\int_0^t\left(
        DV_2(\tau) \,+\, \int_0^\tau \gamma(\tau-s)V_2(s)\ds\right)\dtau\\[1ex]
   &\quad 
    \,=\, \int_0^t \int_0^\tau r'(\tau-\tau')G\dW(\tau') \dtau
    \,=\, \int_0^t \int_{\tau'}^t r'(\tau-\tau') \dtau\, G\dW(\tau')\\[1ex]
   &\quad
    \,=\, \int_0^t r(t-\tau') G\dW(\tau') \,-\, \int_0^t G\dW(\tau') 
    \,=\, V_2(t) \,-\, GW(t)\,.
\end{align*}
This establishes \req{V2-strongsol}, i.e., that $V_2$ is a strong solution
of \req{V2-GLE}. Accordingly, both identities in \req{V12-GLE} hold true,
which was to be shown.

The uniqueness of the solution of \req{GLE} follows from the fact that the
difference of any two solutions is (almost surely) a solution of the
deterministic equation~\req{DiffResolv} with homogeneous initial data, 
and this initial value problem has no nontrivial solutions due
to the uniqueness of the differential resolvent.
\end{proofof}

Next comes the generalized Langevin equation with infinite time horizon.

\begin{proofof}{Theorem~\ref{Thm:GLEinf}}
To begin with, take note that since $r$ is assumed to belong to $L^1(\R^+)$, 
so does $r'$, because \req{DiffResolv} implies that
\bdm
   \norm{r'}_{L^1(\R^+)} \,\leq\, 
   \norm{D}_2\norm{r}_{L^1(\R^+)} \,+\, \norm{\gamma}_{L^1(\R^+)}\norm{r}_{L^1(\R^+)}\,,
\edm
where the $L^1$-norm is defined as
\bdm
   \norm{r}_{L^1(\R^+)} \,=\, \int_0^\infty \norm{r(t)}_2 \dt\,.
\edm
It follows that $r$ is continuous with a well-defined limit
at infinity, and this limiting value must be the zero matrix, for otherwise
$r$ would not belong to $L^1(\R^+)$. Accordingly, $r$ is bounded, and
hence, $r\in L^2(\R^+)$.
It follows that
\begin{subequations}
\label{eq:V1V2inf}
\be{V1inf}
   V_{1,\infty}(t) \,=\, \int_{-\infty}^t r(t-\tau)F_0(\tau)\dtau
\ee
and
\be{V2inf}
   V_{2,\infty}(t) \,=\, \int_{-\infty}^t r(t-\tau)G\dW(\tau)
\ee
are well-defined mean-square continuous bounded Gaussian processes.
\end{subequations}

As in the proof of Theorem~\ref{Thm:GLE} one can now show that 
$V_\infty=V_{1,\infty}+V_{2,\infty}$ is a strong solution of \req{GLEinf}. 
The details of the first step of this argument, 
which consists in proving that $V_{1,\infty}$ has a mean-square continuous
derivative with
\be{V1infty-sol}
   V_{1,\infty}'(t) \,=\, -DV_{1,\infty}(t)
      \,-\, \int_{-\infty}^t \gamma(t-s)V_{1,\infty}(s)\ds \,+\, F_0(t)
\ee
can be omitted; the argument is exactly the same as in the proof of 
Theorem~\ref{Thm:GLE}.

To establish that $V_{2,\infty}$ is a strong solution of 
\be{V2infty-strongsol}
   \dot{V}_{2,\infty}(t) \,=\, -DV_{2,\infty}(t) 
      \,-\, \int_{-\infty}^t \gamma(t-s)V_{2,\infty}(s) \ds \,+\, G\dot{W}(t)\,,
\ee
one can follow the proof of Theorem~\ref{Thm:GLE} to show that
\bdm
   DV_{2,\infty}(t) \,+\, \int_{-\infty}^t \gamma(t-s)V_{2,\infty}(s)\ds
   \,=\, -\int_{-\infty}^t r'(t-\tau)G\dW(\tau)\,.
\edm
Accordingly, the integral of the right-hand side of 
\req{V2infty-strongsol} over any given bounded time interval $[t_1,t_2]$ equals
\begin{align*}
   &\int_{t_1}^{t_2}\left(-DV_{2,\infty}(t) 
       \,-\, \int_{-\infty}^t \gamma(t-s)V_{2,\infty}(s) \ds\right)\!\dt
    \,+\, \int_{t_1}^{t_2} G\dW(t)\\[1ex]
   &\quad
    \,=\, \int_{t_1}^{t_2}\int_{-\infty}^t r'(t-\tau)G\dW(\tau)\dt
          \,+\, \int_{t_1}^{t_2} G\dW(t)\\[1ex]
   &\quad
    \,=\, \int_{-\infty}^{t_2} 
             \left(\int_{\max\{\tau,t_1\}}^{t_2}r'(t-\tau)\dt\right) G\dW(\tau)
          \,+\, \int_{t_1}^{t_2} G\dW(t)\\[1ex]
   &\quad
    \,=\, \int_{-\infty}^{t_2} r(t_2-\tau)G\dW(\tau)
          \,-\, \int_{-\infty}^{t_1} r(t_1-\tau)G\dW(\tau)
          \,-\, \int_{t_1}^{t_2} r(0)G\dW(\tau)\\[1ex]
   &\quad\phantom{\,=\ }   \,+\, \int_{t_1}^{t_2} G\dW(t)\\[2ex]
   &\quad
    \,=\, V_{2,\infty}(t_2) \,-\, V_{2,\infty}(t_1)\,, 
\end{align*}
where in the final step $r(0)=I$ has been used. Accordingly, $V_{2,\infty}$ is
a strong solution of \req{V2infty-strongsol}. Together with \req{V1infty-sol}
this shows that $V_\infty$ is a strong solution of \req{GLEinf}.

Concerning the uniqueness issue, the difference of any two bounded
solutions of \req{GLEinf} is (almost surely)
a bounded solution of the deterministic equation
\bdm
   v'(t) \,=\, - Dv(t) \,-\, \int_{-\infty}^t \gamma(t-s)v(s)\ds\,, \qquad
   t\in\R\,.
\edm
By virtue of \cite[Theorem~3.3.10]{GLS90} -- applied to the full line 
convolution, in which the convolution kernel $\gamma$ is extended by zero to
the negative time axis -- $v=0$ is the only bounded solution of this equation,
and hence, this settles uniqueness.

It remains to prove that $V_\infty(t)-V(t)$ converges to zero in the 
mean-square sense as $t\to\infty$. For a given path of the forcing process $F$ 
it follows from \req{pathwise} and \req{Vinfty} that
\begin{align*}
   V_\infty(t) \,-\, V(t) 
   &\,=\, \int_{-\infty}^0 r(t-s) F(s) \ds - r(t)V_0 \\[1ex]
   &\,=\, \int_{-\infty}^0 r(t-s)F_0(s)\ds \,+\, \int_{-\infty}^0 r(t-s)G\dW(s)
          \,-\, r(t)V_0\,.
\end{align*}
Accordingly,
\bdm
   \E\bigl(\norm{V_\infty(t) \,-\, V(t)}_2^2\bigr)
   \,\leq\, C\left(\left(\int_t^\infty \norm{r(\tau)}_2\dtau\right)^2
            \,+\, \int_t^\infty \norm{r(\tau)}_2^2\dtau
            \,+\, \norm{r(t)}_2^2\right)
\edm 
for some constant $C>0$, and the right-hand side goes to zero as $t\to\infty$
because $r\in L^1(\R^+)\cap L^2(\R^+)$, and $r(t)\to 0$ as $t\to\infty$.
\end{proofof}

\section{Proofs of Theorem~\ref{Thm:fdt} and its corollary}
\label{App:B}
This appendix contains the proof of the second fluctuation-dissipation theorem
for the generalized Langevin equation formulated as an 
initial value problem.

\begin{proofof}{Theorem~\ref{Thm:fdt}}
Let $V_0\sim \Normal(0,\Sigma)$ and let $F$ be as in Assumption~\ref{Ass:F} 
and independent of $V_0$.
Then a unique solution of \req{GLE} exists according to Theorem~\ref{Thm:GLE},
and it can be written in the form $V=V_1+V_2$ with $V_1$ and $V_2$ of \req{V12}.
From this representation it is obvious that $V$ is a centered Gaussian process.

Furthermore, for $s,t\geq 0$ there holds
\begin{align*}
   &V_2(s)V_1(s+t)^* \\[1ex]
   &\quad \,=\, \int_0^s r(s-\tau)G V_0^*r(s+t)^* \dW(\tau) \\[1ex]
   &\quad\phantom{\,=\ } 
                \,+\, \int_0^s\int_0^{s+t} 
                         r(s-\tau)G F_0(\tau')^*r(s+t-\tau')^*
                      \dtau'\dW(\tau)\\[1ex]
   &\quad \,=\, \int_0^s r(s-\tau)GV_0^*r(s+t)^* \dW(\tau)\\[1ex]
   &\quad\phantom{\,=\ } 
                \,+\, \int_0^{s+t}
                      \left(\int_{-\infty}^\infty\int_0^s
                         r(s-\tau)G 
                      \dW(\tau)\dWtilde(\sigma)^*\varphi(\tau'-\sigma)^*
                      \right)\!
                      r(s+t-\tau')^*
                      \dtau'\,,
\end{align*}
and hence,
\begin{subequations}
\label{eq:EV1V2}
\be{EV1V2a}
   \E\bigl(V_1(s+t)V_2(s)^*\bigr)
   \,=\, \int_0^{s+t} \int_0^s 
            r(s+t-\tau')\varphi(\tau'-\tau)G^*r(s-\tau)^*
         \dtau\dtau'
\ee
when $\Wtilde=W$, while
\be{EV1V2b}
   \E\bigl(V_1(s+t)V_2(s)^*\bigr) \,=\, 0
\ee
\end{subequations}
when $\Wtilde$ and $W$ are independent of each other.
In the same manner one obtains
\begin{align*}
   &V_2(s+t)V_1(s)^* \\[1ex]
   &\quad \,=\, \int_0^{s+t} r(s+t-\tau)GV_0^*r(s)^* \dW(\tau) \\[1ex]
   &\quad\phantom{\,=\ } 
                \,+\, \int_0^{s+t}\int_0^s
                         r(s+t-\tau)G F_0(\tau')^*r(s-\tau')^*
                      \dtau'\dW(\tau)\\[1ex]
   &\quad \,=\, \int_0^{s+t} r(s+t-\tau)GV_0^*r(s)^* \dW(\tau) \\[1ex]
   &\quad\phantom{\,=\ } 
                \,+\, \int_0^s\left(\int_{-\infty}^\infty\int_0^{s+t}
                         r(s+t-\tau)G \dW(\tau)\dWtilde(\sigma)^* 
                         \varphi(\tau'-\sigma)^*\right)r(s-\tau')^*\dtau'\,,
\end{align*}
which yields
\begin{subequations}
\label{eq:EV2V1}
\be{EV2V1a}
   \E\bigl(V_2(s+t)V_1(s)^*\bigr)
   \,=\, \int_0^{s+t} \int_0^s
            r(s+t-\tau)G\,\varphi(\tau'-\tau)^*r(s-\tau')^*\dtau'\dtau
\ee
when $\Wtilde=W$, whereas
\be{EV2V1b}
   \E\bigl(V_2(s+t)V_1(s)^*\bigr) \,=\, 0
\ee
\end{subequations}
when $\Wtilde$ and $W$ are independent of each other.
Further, it is readily seen that
\be{EV2V2}
   \E\bigl(V_2(s+t)V_2(s)^*\bigr)
   \,=\, \int_0^s r(s+t-\tau)GG^*r(s-\tau)^*\dtau\,,
\ee
and
\be{EV1V1}
\begin{aligned}
   &\E\bigl(V_1(s+t)V_1(s)^*\bigr) \\
   &\quad \,=\, r(s+t)\Sigma r(s)^*
          \,+\, \int_0^{s+t} \int_0^s 
                   r(s+t-\tau) C_{F_0}(\tau-\tau') r(s-\tau')^*\dtau'\dtau\,.
\end{aligned}
\ee

It follows from \req{EV1V2}--\req{EV1V1} that
\begin{align*}
   \E\bigl(V(s+t)V(s)^*\bigr)
   &\,=\, r(s+t)\Sigma r(s)^* \,+\, \int_0^s r(s+t-\tau)GG^*r(s-\tau)^*\dtau
          \\[1ex]
   &\phantom{\,=\ }
          \,+\, \kappa\int_0^{s+t}\int_0^s 
                   r(s+t-\tau')\varphi(\tau'-\tau)G^* r(s-\tau)^*
                \dtau\dtau'\\[1ex]
   &\phantom{\,=\ }
          \,+\, \kappa\int_0^{s+t}\int_0^s 
                   r(s+t-\tau)G\,\varphi(\tau'-\tau)^* r(s-\tau')^*
                \dtau'\dtau\\[1ex]
   &\phantom{\,=\ }
          \,+\, \int_0^{s+t}\int_0^s 
                   r(s+t-\tau)C_{F_0}(\tau-\tau') r(s-\tau')^*
                \dtau'\dtau\,,
\end{align*}
where $\kappa=1$, if $\Wtilde=W$, whereas $\kappa=0$, 
if $\Wtilde$ and $W$ are two independent Brownian motions.
Together with \req{CF0} this implies that
\be{CV-tmp}
\begin{aligned}
   \E\bigl(V(s+t)V(s)^*\bigr)
   &\,=\, r(s+t)\Sigma r(s)^* \,+\, \int_0^s r(s+t-\tau)GG^*r(s-\tau)^*\dtau
          \\[1ex]
   &\phantom{\,=\ }
          \,+\, \int_0^{s+t}\int_0^s 
                   r(s+t-\tau)\chi(\tau-\tau') r(s-\tau')^*
                \dtau'\dtau\,,
\end{aligned}
\ee
where
\be{chi}
   \chi \,=\, \begin{cases}
                 \varphi G^* \,+\, G\phitrev^* \,+\, \varphi*\phitrev^* \,, &
                 \text{if \ $\Wtilde = W$}\,, \\
                 \varphi*\phitrev^*\,, & 
                 \text{if \ $W$ and $\Wtilde$ are independent}\,.
              \end{cases}
\ee

Consider now the double integral
\be{double}
   \dI \,=\, \int_0^{s+t}\int_0^s 
                r(s+t-\tau)\gamma_\Sigma(\tau-\tau') r(s-\tau')^*
             \dtau'\dtau\,,
\ee
where $\gamma_\Sigma$ is defined as in \req{gamma-Sigma}.
For $0\leq\tau'\leq s$ and $\tau'\leq\tau\leq s+t$ the argument $\tau-\tau'$
of $\gamma_\Sigma$ in \req{double} is nonnegative. 
The corresponding part ${\cal I}_+$ of the double integral~\req{double} 
can thus be rewritten by Fubini's theorem as
\bdm
   {\cal I}_+ \,=\, 
   \int_0^s \left(
               \int_0^{s+t-\tau'} r(s+t-\tau'-\tau'') \gamma(\tau'')\dtau''
           \right) 
           \Sigma\,r(s-\tau')^*\dtau'\,.
\edm
Replacing the inner integral by means of \req{DiffResolv2} one therefore 
arrives at
\be{Iplus}
   {\cal I}_+
   \,=\, -\int_0^s r(s+t-\tau)D\Sigma\,r(s-\tau)^*\dtau
         \,-\, \int_0^s r'(s+t-\tau)\Sigma\,r(s-\tau)^*\dtau\,.
\ee
In the remaining part of the domain of integration of \req{double}, 
where $0\leq\tau'\leq s$ and $0\leq\tau<\tau'$, 
the argument $\tau-\tau'$ of $\gamma_\Sigma$ is negative. 
Therefore, using \req{gamma-Sigma} and Fubini's theorem,
the corresponding part ${\cal I}_-$ of $\dI$ becomes
\begin{align*}
   {\cal I}_-
   &\,=\, \int_0^s \int_0^{\tau'}
             r(s+t-\tau)\Sigma\gamma(\tau'-\tau)^*r(s-\tau')^*
          \dtau\dtau'\\[1ex]
   &\,=\, \int_0^s r(s+t-\tau)\Sigma \left(
             \int_0^{s-\tau} r(s-\tau-\tau'')\gamma(\tau'')
             \dtau''\right)^{\!*}\!
          \dtau\,,
\end{align*}
and together with \req{DiffResolv2} it follows that
\be{Iminus}
   {\cal I}_-
   \,=\, -\int_0^s r(s+t-\tau)\Sigma D^*r(s-\tau)^*\dtau
         \,-\, \int_0^s r(s+t-\tau)\Sigma\,r'(s-\tau)^*\dtau\,.
\ee
Using \req{Iplus} and \req{Iminus}, \req{double} thus yields
\begin{align*}
   \dI &\,=\, {\cal I}_+ + {\cal I}_-\\[1ex]
   &\,=\, - \int_0^s r(s+t-\tau)(D\Sigma+\Sigma D^*)r(s-\tau)^*\dtau \\[1ex]
   &\phantom{\,=\ } \,-\, \int_0^s r'(s+t-\tau)\Sigma\,r(s-\tau)^*\dtau
          \,-\, \int_0^s r(s+t-\tau)\Sigma\,r'(s-\tau)^*\dtau\\[1ex]
   &\,=\, - \int_0^s r(s+t-\tau)(D\Sigma+\Sigma D^*)r(s-\tau)^*\dtau \\[1ex]
   &\phantom{\,=\ }
          \,+\, \int_0^s \!\!
                   \frac{\rmd}{\dtau} \Bigl(r(s+t-\tau)\Sigma\,r(s-\tau)^*\Bigr)
                \!\dtau\\[1ex]
   &\,=\, - \int_0^s r(s+t-\tau)(D\Sigma+\Sigma D^*)r(s-\tau)^*\dtau 
          \,+\, r(t)\Sigma \,-\, r(s+t)\Sigma\,r(s)^*\,;
\end{align*}
recall that $r(0)=I$.

Adding this result to \req{CV-tmp} gives
\begin{align*}
   &\E\bigl(V(s+t)V(s)^*\bigr)\\[1ex]
   &\quad
    \,=\, r(t)\Sigma \,-\, \dI  
          \,+\, \int_0^s r(s+t-\tau)(GG^*-D\Sigma-\Sigma D^*)r(s-\tau)^*\dtau
          \\[1ex]
   &\quad\phantom{\,=\,}
          \,+\, \int_0^{s+t}\int_0^s 
                   r(s+t-\tau)\chi(\tau-\tau') r(s-\tau')^*
                \dtau'\dtau\,,
\end{align*}
and inserting \req{double} for $\dI$ this can be rewritten as
\be{preliminary}
\begin{aligned}
   &\E\bigl(V(s+t)V(s)^*\bigr)\\[1ex]
   &\quad\,=\, r(t)\Sigma 
          \,+\, \int_0^s r(s+t-\tau)(GG^*-D\Sigma-\Sigma D^*)r(s-\tau)^*\dtau
          \\[1ex]
   &\quad\phantom{\,=\,}
          \,+\, 
          \int_0^{s+t}\int_0^s 
             r(s+t-\tau)(\chi-\gamma_\Sigma)(\tau-\tau')
             r(s-\tau')^*
          \dtau'\dtau\,,
\end{aligned}
\ee
valid for $s,t\geq 0$.
Take note that
\bdm
   \chi(-t) \,=\, \chi(t)^* \quad \text{and} \quad
   \gamma_\Sigma(-t) \,=\, \gamma_\Sigma(t)^*\,, \qquad t\in\R\,.
\edm
Therefore, $\chi=\gamma_\Sigma$ a.e.\ on $\R$, once it holds true almost
everywhere on the positive time axis, i.e., 
if the assumption~\req{fdt-relation} is valid.
Accordingly, if \req{GG} and \req{fdt-relation} are satisfied,
then $V$ is a stationary solution of the 
generalized Langevin equation~\req{GLE} by virtue of \req{preliminary}, 
and its autocorrelation function coincides for $t\geq 0$ 
with the associated differential resolvent $r$ times $\Sigma$.
This establishes the if-part of the theorem. 

For the only-if part, assume that the solution $V$ of \req{GLE} is
a stationary process. Then it follows from \req{preliminary}
and Fubini's theorem that the autocorrelation function $C_V$ of $V$ 
can be rewritten in the form  
\be{CVconverse-tmp2}
\begin{aligned}
   C_V(t) 
   &\,=\, \E\bigl(V(s+t)V(s)^*\bigr)\\[1ex]
   &\,=\, r(t)\Sigma 
          \,+\, \int_0^s r(s+t-\tau)(GG^*-D\Sigma-\Sigma D^*)r(s-\tau)^*\dtau
          \\[1ex]
   &\phantom{\,=\ }\phantom{r(t)\Sigma}\      
          \,+\, \int_0^s \Delta(s+t,\tau') r(s-\tau')^*\dtau'  
\end{aligned}
\ee
for every $s,t\geq 0$, with
\be{Delta}
   \Delta(\tau,\tau') \,=\, 
   \int_0^\tau 
      r(\tau-\tau'')(\chi-\gamma_\Sigma)(\tau''-\tau')\dtau''\,, \qquad
   0 \leq \tau' \leq \tau\,.
\ee
In particular, for $s=0$ and every $t\geq 0$ this shows that
\bdm
   C_V(t) \,=\, \E\bigl(V(t)V(0)^*\bigr) \,=\, r(t)\Sigma\,,
\edm
i.e., that \req{CV} holds true, as well as
\be{zero}
   0 \,=\, \!\left.\frac{\!\rmd}{\!\ds}\,\E\bigl(V(s+t)V(s)^*\bigr)\right|_{s=0}
   \!=\, r(t)\bigl(GG^*-D\Sigma-\Sigma D^*)\,+\, \Delta(t,0)\,,
\ee
where, again, $r(0)=I$ has been used in the final step.
Note that $\Delta(0,0)=0$, and hence, 
\bdm
   GG^* \,=\, D\Sigma+\Sigma D^*
\edm
follows from \req{zero} for $t=0$. Accordingly, \req{zero} further implies that
\be{Delta-help}
   0 \,=\, \Delta(t,0) 
     \,=\, \int_0^t r(t-\tau)(\chi-\gamma_\Sigma)(\tau)\dtau 
   \qquad \text{for every $t\geq 0$.}
\ee

Since $r\Sigma$ coincides with the autocorrelation function of $V$ for
nonnegative arguments, the function $r$ is bounded, and therefore
its Laplace transform 
\bdm
   \L r(\zeta) \,=\, \int_0^\infty e^{-\zeta t} r(t)\dt
\edm
is well-defined for $\Real\zeta>0$, and \req{Delta-help} implies that
\be{former}
   0 \,=\, \L r(\zeta)\,\L(\chi-\gamma_\Sigma)(\zeta)\,, \qquad
   \Real\zeta > 0\,.
\ee
On the other hand, the Laplace transform of the 
integro-differential equation~\req{DiffResolv} yields
\bdm
   \zeta\L r(\zeta) \,-\, I 
   \,=\, -D\L r(\zeta) \,-\, \L\gamma(\zeta)\,\L r(\zeta)
\edm
for $\Real\zeta > 0$, showing that
\be{Lr}
   \bigl(\zeta I \,+\, D \,+\, \L\gamma(\zeta)\bigr)\L r(\zeta) \,=\, I\,.
\ee
From the latter identity follows that the null space of $\L r(\zeta)$ 
is trivial for every $\zeta$ with $\Real\zeta > 0$, 
and therefore \req{former} implies that the Laplace transform of 
$\chi-\gamma_\Sigma$ vanishes identically in the open right half plane.
Since the Laplace transform is injective
(cf., e.g., Doetsch~\cite[Section~5]{Doet74}), this shows that
\bdm
   \chi \,=\, \gamma_\Sigma \,=\, \gamma\Sigma \qquad \text{a.e. on $\R^+$}.
\edm
Thus, we have established that \req{fdt-relation} also holds true, 
and the proof is done.
\end{proofof}

\begin{proofof}{Corollary~\ref{Cor:fdt-Kubo}}
In order to apply Theorem~\ref{Thm:fdt} to the 
Langevin equation~\req{GLE-Kubo}, let $D=0$ in \req{GLE}. Further, let
the fluctuating force satisfy Assumption~\ref{Ass:F} with $G=0$. 
Then condition~\req{GG} of Theorem~\ref{Thm:fdt} is satisfied. 
Since $F_0$ in Assumption~\ref{Ass:F} has the representation~\req{F0phi}
for some $\varphi\in L^2(\R)$, 
condition~\req{fdt-relation} is equivalent to \req{Cor:fdt-Kubo}
by virtue of \req{CF0}; in this case
the assertion is therefore an immediate consequence of
Theorem~\ref{Thm:fdt}.


In order to have a representation~\req{F0phi} for some $\varphi\in L^2(\R)$,
however, it is necessary, that $\gamma_\Sigma$ of \req{gamma-Sigma} belongs
to $L^1(\R)$;
compare \req{CF0hat}. 
But a careful inspection of the proof
of Theorem~\ref{Thm:fdt} reveals that the assertion of the theorem remains
valid for $D=G=0$, when the fluctuating force $F=F_0$ is \emph{any} 
mean-square continuous stationary Gaussian process.
Although such a process cannot always be written in the form~\req{F0phi} 
for some $\varphi\in L^2(\R)$, the autocorrelation function of $V$ is still
given by \req{preliminary} with $\chi=C_{F_0}$, 
and since the latter is a bounded function,
its Laplace transform is well-defined for $\Real\zeta>0$.
This is all what is needed to make the proof go through.
\end{proofof}

\section{Proofs of Theorem~\ref{Thm:fdtinf}, 
Corollary~\ref{Cor:fdtinf-D0}, and Proposition~\ref{Prop:CFV}}
\label{App:C}
This appendix contains the proofs of the results on stationary solutions
of the generalized Langevin equation~\req{GLEinf} with infinite time horizon.

\begin{proofof}{Theorem~\ref{Thm:fdtinf}}
Under the assumptions formulated in Theorem~\ref{Thm:fdtinf}
the generalized Langevin equation~\req{GLEinf}
with infinite time horizon has a unique bounded solution $V_\infty$
by virtue of Theorem~\ref{Thm:GLEinf},
which can be decomposed as $V_\infty=V_{1,\infty}+V_{2,\infty}$
with $V_{1,\infty},V_{2,\infty}$ of \req{V1V2inf}.
Let $s\in\R$ and $t\geq 0$: 
Following the proof of Theorem~\ref{Thm:fdt} in \ref{App:B} 
one easily arrives at
\begin{align*}
   \E\bigl(V_\infty(s+t)V_\infty(s)^*\bigr)
   &\,=\, \int_{-\infty}^s r(s+t-\tau)GG^*r(s-\tau)^*\dtau\\[1ex]
   &\phantom{\,=\ }
          \,+\, \int_{-\infty}^{s+t}\int_{-\infty}^s 
                   r(s+t-\tau)\chi(\tau-\tau') r(s-\tau')^*
                \dtau'\dtau
\end{align*}
in place of \req{CV-tmp}, where $\chi$ is defined in \req{chiinf}.
Substituting $\tau$ by $\tau+s$ and $\tau'$ by $\tau'+s$ it follows that
\begin{align*}
   \E\bigl(V_\infty(s+t)V_\infty(s)^*\bigr)
   &\,=\, \int_{-\infty}^0 r(t-\tau)GG^*r(-\tau)^*\dtau\\[1ex]
   &\phantom{\,=\ }
          \,+\, \int_{-\infty}^t\int_{-\infty}^0 
                   r(t-\tau)\chi(\tau-\tau') r(-\tau')^*
                \dtau'\dtau\,,
\end{align*}
which is independent of $s$, and agrees with \req{Thm:fdtinf}.
This shows that $V_\infty$ is a stationary process.

For $\Sigma\in\R^{d\times d}$ symmetric positive definite and $\gamma_\Sigma$
of \req{gamma-Sigma} one can proceed as in the proof of Theorem~\ref{Thm:fdt},
compare~\req{double}, and establish that
\begin{align*}
   &\int_{-\infty}^t \int_{-\infty}^0 
       r(t-\tau)\gamma_\Sigma(\tau-\tau')r(-\tau')^*\dtau'\dtau\\[1ex]
   &\qquad
    \,=\, r(t)\Sigma \,-\, 
          \int_0^\infty r(t+\tau)(D\Sigma+\Sigma D^*)r(\tau)^*\dtau\,, \qquad
    t \geq 0\,.
\end{align*}
Adding this equation to \req{Thm:fdtinf} the alternative representation
\req{Thm:fdtinf2} of $C_{V_\infty}$ is readily obtained.
\end{proofof}

\begin{proofof}{Corollary~\ref{Cor:fdtinf-D0}}
According to \req{Lr} and Remark~\ref{Rem:PW} 
the Fourier transform $\rhat_0$ of 
\bdm
   r_0(t) \,=\, \begin{cases}
                   r(t) \,, & t\geq 0\,, \\
                   0\,, & t<0\,,
                \end{cases}
\edm
can be expressed in terms of the Laplace transform of the memory kernel, namely
\bdm
   \rhat_0(\omega) \,=\, \L r(\rmi\omega)
   \,=\, \bigl(\rmi\omega I + \L\gamma(\rmi\omega)\bigr)^{-1}\,, 
   \qquad \omega\in\R\,.
\edm

The Fourier transform $\psihat$ of $\psi$, on the other hand, satisfies
\bdm
   \psihat(\omega) \,=\, 
   \frac{\rmi}{\omega^3}\,\widehat{\psi'''}(\omega)\,, \qquad 
   \omega\in\R\,,
\edm
and hence, since $\psihat(\omega)$ is positive semidefinite, 
the function $\phihat$ defined by
\bdm
   \phihat(\omega) \,=\, \rhat_0(\omega)^{-1}\psihat(\omega)^{1/2}
   \,=\, \rmi
         \bigl(I\,-\, \frac{\rmi}{\omega}\L\gamma(\rmi\omega)\bigr)
         \bigl(\frac{\rmi}{\omega}\widehat{\psi'''}(\omega)\bigr)^{1/2}
\edm
is a well-defined element of $L^2(\R)$,
because $\widehat{\psi'''}(\omega)/\omega\in L^1(\R)$ by virtue of the 
Cauchy-Schwarz inequality.

Choosing the fluctuating force of \req{GLEinf} according to 
Assumption~\ref{Ass:F} with $G=0$ and the inverse Fourier transform 
$\varphi$ of $\phihat$, the autocorrelation function of the 
stationary solution $V_\infty$ of \req{GLEinf} is given by 
\bdm
   C_{V_\infty}(t)   
   \,=\, \int_0^\infty\int_0^\infty 
            r(\tau)(\varphi*\phitrev^*)(t+\tau'-\tau)r(\tau')^*\dtau'\dtau
\edm
according to Theorem~\ref{Thm:fdtinf}. The convolution theorem therefore yields
\bdm
   \widehat{C}_{V_\infty} \,=\, \rhat_0 \phihat \phihat^*\rhat_0^* \,=\, \psihat\,,
\edm
proving the assertion that $C_{V_\infty}=\psi$.
\end{proofof}

\begin{proofof}{Proposition~\ref{Prop:CFV}}
Starting from the representation $V_\infty=V_{1,\infty}+V_{2,\infty}$ of 
\req{V1V2inf} there holds
\be{CV1F0}
   \E\bigl(V_{1,\infty}(s+t)F_0(s)^*\bigr)
   \,=\, \int_{-\infty}^{s+t} r(s+t-\tau)C_{F_0}(\tau-s)\dtau\,,
\ee
where $C_{F_0}$ has been computed in \req{CF0}.
If $\Wtilde=W$, then
\be{CV2F0}
   \E\bigl(V_{2,\infty}(s+t)F_0(s)^*\bigr)
   \,=\, \int_{-\infty}^{s+t}r(s+t-\tau)G \varphi(s-\tau)^*\dtau\,,
\ee
while $V_{2,\infty}$ and $F_0$ are uncorrelated, when $\Wtilde$ and $W$
are independent Brownian motions.

Since the component $G\dot{W}$ of the fluctuating force is a 
random distribution, rather than a random process, 
one has to use the distributional definition from Remark~\ref{Rem:distribution}
to evaluate the covariance structure of $V_\infty$ and $G\dot{W}$. 
Of course, if $\Wtilde$ and $W$ are independent
of each other, then $V_{1,\infty}$ and $G\dot{W}$ are uncorrelated.
If $\Wtilde=W$, on the other hand, and $f,g\in C_0^\infty(\R)$ are arbitrary
vector-valued test functions, then
\bdmal
   &\scalp{V_{1,\infty},f}\scalp{G\dot{W},g}
    \,=\, \int_{-\infty}^\infty f(t)^*\int_{-\infty}^t r(t-\tau)F_0(\tau)\dtau\dt
          \int_{-\infty}^\infty g(s)^*G\dW(s) \\[1ex]
   &\,=\, \int_{-\infty}^\infty\int_{-\infty}^\infty\int_{-\infty}^t 
             f(t)^*r(t-\tau)\varphi(\tau-\tau')\dtau\dt\dW(\tau')
          \int_{-\infty}^\infty g(s)^*G\dW(s) \,,
\edmal
and hence,
\bdm
   \E\bigl(\scalp{V_{1,\infty},f}\scalp{G\dot{W},g}\bigr)
   \,=\, \int_{-\infty}^\infty\int_{-\infty}^\infty f(t)^*
            \int_{-\infty}^t r(t-\tau) \varphi(\tau-s)G^*g(s)\dtau\dt\ds\,,
\edm
which means that the corresponding correlation distribution can be identified
with the function
\be{CV1GW}
   \E\bigl(V_{1,\infty}(s+t)\bigl(G\dot{W}(s)\bigr)^*\bigr)
   \,=\, \int_{-\infty}^{s+t} r(s+t-\tau) \varphi(\tau-s)G^*\dtau\,,
\ee
when $\Wtilde=W$. Adding up the identities~\req{CV1F0}-\req{CV1GW}
it follows that
\begin{align}
\nonumber
   \E\bigl(V(s+t)F(s)^*\bigr)
   &\,=\, \int_{-\infty}^{s+t} r(s+t-\tau)\chi(\tau-s)\dtau  
          \,+\, \E\bigl(V_{2,\infty}(s+t)\bigl(G\dot{W}(s)\bigr)^*\bigr)\\[1ex]
\label{eq:Prop:CFV-tmp}
   &\,=\, \int_0^\infty r(\tau) \chi(t-\tau)\dtau 
          \,+\, \E\bigl(V_{2,\infty}(s+t)\bigl(G\dot{W}(s)\bigr)^*\bigr)
\end{align}
with $\chi$ of \req{chiinf}.

It remains to compute the correlation structure between $V_{2,\infty}$ and
$G\dot{W}$. Again, let $f,g\in C_0^\infty(\R)$ be given test functions.
Then Fubini's theorem gives
\bdm
   \scalp{V_{2,\infty},f}\scalp{G\dot{W},g}
   \,=\, \int_{-\infty}^\infty f(t)^*\int_{-\infty}^t r(t-\tau)G\dW(\tau)
          \int_{-\infty}^\infty g(s)^*G\dW(s)\dt \,,
\edm
so that
\bdm
   \E\bigl(\scalp{V_{2,\infty},f}\scalp{G\dot{W},g}\bigr)
   \,=\, \int_{-\infty}^\infty f(t)^*\int_{-\infty}^t r(t-s)GG^*g(s)\ds\dt\,,
\edm
showing that 
\bdm
   \E\bigl(V_{2,\infty}(s+t)\bigl(G\dot{W}(s)\bigr)^*\bigr)
   \,=\, \begin{cases}
            r(t)GG^*\,, & t>0\,,\\
            0\,, & t<0\,.
         \end{cases}
\edm
Inserting this result into \req{Prop:CFV-tmp} finally yields the assertion.
\end{proofof}

\section*{Acknowledgement}
The author is grateful to Friederike Schmid, Niklas Bockius, and Max Braun
for their careful reading of a preliminary version of this paper
and for various suggestions which helped to improve the presentation.


\end{document}